\documentclass[11pt]{article}
\usepackage{amsmath,paralist,amsthm,comment,amssymb}
\usepackage{graphicx}
\usepackage[parfill]{parskip}
\usepackage{tikz}
\usetikzlibrary{matrix,arrows}
\usepackage{cite}
\usepackage{fullpage}

\usepackage{algorithm} 
\usepackage{algorithmic}

\renewcommand{\>}{\rangle}

\newcommand{\ket}[1]{|#1\rangle}
\newcommand{\mc}[1]{\mathcal{#1}}
\newcommand{\mbb}[1]{\mathbb{#1}}

\newcommand{\C}{\mathbb{C}}
\newcommand{\R}{\mathbb{R}}
\newcommand{\Z}{\mathbb{Z}}

\newcommand{\rep}{\rm{Rep}}
\newcommand{\acc}{\rm{acc}}
\newcommand{\sinc}{\rm{sinc}}

\DeclareMathOperator{\bs}{\rm{bs}}
\DeclareMathOperator{\gs}{\rm{gs}}
\DeclareMathOperator{\im}{\rm{im}}
\newcommand{\cA}{\mathcal{A}}
\newcommand{\cB}{\mathcal{B}}
\newcommand{\cG}{\mathcal{G}}

\newcommand{\nix}[1]{}
\newcommand{\floor}[1]{{\left\lfloor #1\right\rfloor}}
\newcommand{\ceil}[1]{{\left\lceil #1\right\rceil}}
\newcommand{\spn}[1]{{\left\langle #1\right\rangle}}
\newcommand{\fc}[1]{{\left\lfloor #1\right\rceil}}
\newcommand{\be}{\begin{eqnarray*}}
\newcommand{\ee}{\end{eqnarray*}}
\newcommand{\ben}{\begin{eqnarray}}
\newcommand{\een}{\end{eqnarray}}

\newcommand{\ba}{\begin{array}}
\newcommand{\ea}{\end{array}}

\newtheorem{theorem}{Theorem}
\newtheorem{lemma}[theorem]{Lemma}
\newtheorem{prop}[theorem]{Proposition}
\newtheorem{corollary}[theorem]{Corollary}

\newtheorem{example}{Example}
\newtheorem{defn}{Definition}

\begin{document}
\title{\bf Quantum Algorithms for One-Dimensional Infrastructures}
\author{Pradeep Sarvepalli\thanks{School of Chemistry and Biochemistry, Georgia Institute of Technology, Atlanta, GA~30318, USA. Email:\texttt{pradeep.sarvepalli@gatech.edu}. This work was done when P.S. was at the University of British Columbia, Vancouver.} \quad    Pawel Wocjan\thanks{Department of Electrical Engineering and Computer Science, University of Central Florida, Orlando, FL~32816, USA. Email:
        \texttt{wocjan@eecs.ucf.edu}}} 
\date{May 30, 2012} 

\maketitle
\begin{abstract}
Infrastructures are group-like objects that make their appearance in arithmetic geometry in the study of computational problems related to number fields and function fields over finite fields. The most prominent computational tasks of infrastructures are the computation of the circumference of the infrastructure and the generalized discrete logarithms.  Both these problems are not known to have efficient classical algorithms for an arbitrary infrastructure.  Our main contributions are polynomial time quantum algorithms for one-dimensional infrastructures that satisfy certain conditions.  For instance, these conditions are always fulfilled for infrastructures obtained from number fields and function fields, both of unit rank one.  Since quadratic number fields give rise to such infrastructures, this algorithm can be used to solve  Pell's equation and the principal ideal  problem.  In this sense we generalize Hallgren's quantum algorithms for quadratic number fields, while also providing a polynomial speedup over them. Our more general approach shows that these quantum algorithms can also be applied to infrastructures obtained from complex cubic and totally complex quartic number fields.  Our improved way of analyzing the performance makes it possible to show that these algorithms succeed with constant probability independent of the problem size.   In contrast, the lower bound on the success probability due to Hallgren decreases as the fourth power of the logarithm of the circumference.  Our analysis also shows that fewer qubits are required. We also contribute to the study of infrastructures, and show how to compute efficiently within infrastructures.

\end{abstract}

\maketitle
\newpage
\tableofcontents

\section{Introduction}
One of the most important challenges in quantum computing has been the task of finding efficient
algorithms for problems that are intractable on a classical computer. Following Shor's
discovery of a polynomial time quantum algorithm for factoring integers and solving the discrete logarithm problem \cite{shor97}, the key ideas of the period finding 
algorithm were generalized and led to the framework of the hidden subgroup problem (HSP) \cite{Jozsa01}.  The major algorithmic success in this context is that the abelian HSP can be solved efficiently
by a quantum algorithm (while classical algorithms are  inefficient).  This quantum algorithm can also be viewed as determining the structure of a hidden lattice $\Lambda$ inside $\Z^n$.

An important restriction of this quantum algorithm is that it only works for integral lattices.
But,  Hallgren overcame this obstacle in the one-dimensional setting by generalizing Shor's period finding 
algorithm to the case where the period is irrational \cite{Hallgren02,hallgren07} (see also \cite{jozsa03,schmidt09}).  This enabled him to give polynomial time quantum algorithms for 
computing the regulator of a quadratic number field and solving the principal ideal problem.  Schmidt and Vollmer \cite{SV04,SV05} and Hallgren \cite{Hallgren05} presented 
a polynomial time quantum algorithm for determining a hidden lattice in $\R^n$ for fixed $n$.  They showed that computing the unit group and solving the principal ideal  problem in number
fields of fixed unit rank can be solved efficiently with this algorithm.\footnote{Hallgren also showed in \cite{Hallgren05}  how to compute the class group of a number field of fixed unit rank.}  In stark contrast to $\Z^n$, 
the success probability of the above quantum algorithms for finding a hidden lattice in $\R^n$ decreases exponentially  with the dimension, making them 
inefficient with respect to the dimension.    
Thus, an important open problem is to determine whether there exist quantum algorithms whose success probability decrease less rapidly with the dimension.

In this paper, we   initiate the study of  quantum algorithms for  infrastructures. These group-like structures  are hidden beneath the number theoretic details of the above  quantum algorithms. 
They play an important role in the research on computational problems in global fields, i.e. number fields
and function fields over finite fields \cite{fontein09} (arithmetic geometry provides a unified treatment of global fields \cite{Lorenzini96}).
For instance, computing the unit group and solving the principal ideal problem can both be translated to well defined problems of infrastructures, namely, the computation of the lattice
characterizing the periodic symmetry of the infrastructure and the computation of generalized discrete logarithms in these group-like structures.  Both these computational problems associated with the infrastructures are not known to have 
efficient classical algorithms. 

In this paper we focus on arbitrary one-dimensional infrastructures and give polynomial time quantum algorithms
for computing the circumference and for computing the generalized discrete logarithms. 
One-dimensional infrastructures arise from global fields of unit rank, and include the special case of real quadratic number fields studied by Hallgren \cite{hallgren07} and complex cubic and quartic number fields \cite{BW88}, thereby providing further applications. Our algorithms perform better than 
the algorithms of \cite{hallgren07} when applied to these problems. 
The proposed algorithms provide a super polynomial speedup over classical algorithms. 
In addition, we make several other contributions. 
Firstly, although our algorithms are given in a more general setting, they have lower complexity and a higher success probability than those in \cite{Hallgren02,hallgren07}. In fact, all our algorithms can be shown to have a success probability that is lower bounded by a constant, which is independent of the problem size. For instance, our analysis shows that the success probability of computing the circumference is a constant and at least $10^{-5}$, in contrast to \cite{hallgren07} which implies a lower bound  less than $10^{-9}$ and decreases as a fourth power of the circumference. 
It is also better than the result of \cite{schmidt09} which is lower bounded by $2^{-26}$. 
 Secondly, our results when specialized to quadratic
number fields provide  a simpler treatment of the computational problems, and can be 
easily applied without extensive knowledge of number theory.  Thirdly, we introduce an
interesting technical result that could have wider applicability in the analysis of
quantum algorithms employing  quantum Fourier transform. Finally, we make a contribution to the
study of one-dimensional infrastructures by showing how to perform finite precision computations efficiently 
within the infrastructures.  These are useful even in the context of purely classical 
algorithms for infrastructures. 
A natural direction for further investigation is the generalization of the proposed quantum 
algorithms for higher dimensional infrastructures.  
These are presented in \cite{fontein11}.

This paper is structured as follows. We first introduce the mathematical preliminaries,
defining precisely the notion of an infrastructure and the computational problems associated 
with them. We then show that these infrastructures can be endowed with a group structure 
and review the relevant results related to the embedding of the infrastructures into 
circle groups.  We then introduce group homomorphisms that are key to solving the
computational problems associated to them. We also show that these homomorphisms can be 
computed efficiently. These results should be of interest beyond the present context.

In section~\ref{sec:periodicStates}, we generalize the notion of periodic quantum states and prove a key technical
result related to the analysis of Fourier sampling. This result simplifies the analysis 
of the algorithms and leads to   a tighter bound on the success probabilities of 
the proposed algorithms. In this section, we give a quantum algorithm for estimating the 
period of a pseudo-periodic quantum state. This result could be applicable to situations beyond the 
current setting of infrastructures. 

In section~\ref{sec:circumference}, we show how to set up periodic quantum states from 
infrastructures and use the quantum algorithm proposed in section~3 to estimate the 
circumference of the infrastructure. 
In section~\ref{sec:dlog}, we present the  quantum algorithm to solve the generalized discrete 
logarithm
problem. 
%
%

\section{Infrastructures}\label{sec:infra}

We define infrastructures and state the two main computational problems associated to infrastructures. We restrict our attention to the one-dimensional infrastructures. 

\subsection{Definition of infrastructures} 

We refer the reader to \cite{fontein08,fontein08a,fontein09} for 
more information on infrastructures.  Our presentation follows \cite{fontein08,fontein09a}.

\begin{defn}[Infrastructure]
An infrastructure of circumference $R$ is a  pair $(X,d)$ where $X$ is a finite set and  
$d:X\hookrightarrow \mathbb{R}/R\mbb{Z}$ an injective function on $X$.
\end{defn}
Injectivity of $d$ ensures that 
no two distinct elements of $X$ have the same distance. 
\nix{
This ordering induces naturally a function on $X$ such that the orbit
of any element is $X$. This function, called the 
baby-step, $\bs: X \rightarrow X $ is defined as follows.
\ben
\bs(x_i) =\left\{\ba{ll}x_{i+1} & 0\leq i<m-1  \\ x_0& i=m-1 \ea \right. \label{eq:bs}
\een
}
We define a function on the set $X$ called the baby-step, 
$\bs: X \rightarrow X $ as follows. Consider the following set 
\ben
S_x = \{r\in \mbb{R}\mid r> 0 \text{ and } d(x)+r \bmod R \in d(X)\}.
\een 
Let $f_x=\min S_x$.  Then $\bs(x)=x'$ such that $d(x')=d(x)+f_x \mod R $. We also define 
the relative distance function
\ben
\Delta_{\bs} : X \rightarrow \mbb{R}  \text{  where } \Delta_{\bs}(x)=f_x=\min S_x.
\een
Informally, the $\bs(x)$ gives the element next to $x$. The circumference of the infrastructure, denoted $R$, can be expressed in terms of this relative
distance function as follows:
\ben
R= \sum_{i=0}^{m-1}\Delta_{\bs}(x_i).\label{eq:circum}
\een
It is clear that $\bs^{-1}$, the inverse of $\bs$, is well-defined. 
Further, a group-like structure
is imposed on the set $X$ by means of a binary operator, called the giant-step. 
Consider the set 
$$
S_{x,y} = \{ r\in \mbb{R}\mid r\geq 0 \text{  and } d(x)+d(y) +r \bmod R \in d(X) \}.
$$
Let $f_{x,y}=\min S_{x,y}$. Then 
 $\gs : X\times X \rightarrow X$ is defined as:
\ben
\gs(x,y) =  z \text{ such that } d(z) = d(x)+d(y)+f_{x,y} \bmod R.
\een
We define the relative distance function $\Delta_{\gs}$ as: 
\ben 
\Delta_{\gs} : X\times X \rightarrow \mbb{R} \text{ where } \Delta_{\gs}(x,y)= f_{x,y} =\min S_{x,y}.
\een

The giant-step is commutative, but not associative. It is ``almost associative'' in the sense that  for two arbitrary elements $x,y\in X$ the giant-step gives an element $z\in X $ whose distance satisfies $d(z)\approx d(x)+d(y)$.

In infrastructures arising out of quadratic number fields the elements of the 
infrastructure correspond to the principal reduced ideals of the number field. The distance 
function is the norm of the ideals. One can cycle through these ideals using the so-called reduction operator \cite{jozsa03}; this function corresponds to 
the baby-step. One can also define the product of ideals which after reduction
corresponds to the giant-step, see \cite{jozsa03}. 

The definitions of $\bs$ and $\gs$ and the relative distance functions $\Delta_{\bs}$ and $\Delta_{\gs}$ may suggest that we need $R$ and the distance function $d$ to be able to compute them.  However, this is not the case.  These functions can be computed efficiently without the knowledge of $R$ or the distance function $d$.  To illustrate this point, let us explain how (discrete) infrastructures can be considered as generalizations of finite cyclic groups.  

\medskip
\begin{defn}[Discrete infrastructure]
An infrastructure is said to be discrete if its circumference $R$ is a positive integer and its distance function $d$ is integer-valued, i.e., 
$d : X \hookrightarrow \mathbb{Z} / R \mathbb{Z}$. 
\end{defn}

\medskip
\begin{example}[Finite cyclic group]
Suppose $G=\spn{g}$ is a finite cyclic group of order $R$ and generated by $g$. Then we can
form an infrastructure out of $G$ as follows. We let $X=G$ and   define
$d(h)=\log_g h $, for any $h\in G$,
since every element $h\in G$  is of the form $g^{d(h)}$ for some $d(h)\in \mbb{Z}$.
 The baby step $\bs$ of the infrastructure corresponds simply to multiplication of elements $x$ by the generator $g$, 
while the giant step $\gs$ corresponds to the multiplication of two elements $x$ and $y$ in $G$.  The relative distance functions $\Delta_{\bs}$ and $\Delta_{\gs}$ are constant and take on the
values $1$ and $0$, respectively. 
\end{example}

We can now interpret the order of $G$ as the circumference of the infrastructure. The distance function $d(x)$ corresponds to the discrete logarithm of the element $x$ with
respect to the base $g$.  This example makes it clear why we cannot necessarily determine the circumference and the distance function efficiently, even though we can efficiently evaluate the baby and giant steps and their corresponding distance functions.   

%
%

\subsection{Computational problems} 
The main computational problems related to infrastructures are the computation of the
circumference and the computation of generalized discrete logarithms. 

We consider only infrastructures that satisfy the assumptions below.  These are necessary to be able to carry out basic arithmetic operations in infrastructures
in polynomial time. The cost is measured with respect to the input problem size $n$.

\begin{compactenum}[\textbf{A}1)]
\item The circumference satisfies $R \le 2^{\text{poly}(n)}$.
\item Any element $x\in X$ can be represented by a bit string of length $\text{poly}(n)$.
\item The elements $\bs(x)$, $\bs^{-1}(x)$, $\gs(x,y)$ can be determined in time $\text{poly}(n)$ for all $x,y\in X$.  
\item The relative distances $\Delta_{\bs}(x)$ and $\Delta_{\gs}(x,y)$ cannot necessarily be computed exactly.  We only obtain approximate values $\tilde{\Delta}_{\bs}(x)$ and $\tilde{\Delta}_{\bs}(x,y)$ with 
\begin{equation}
| \Delta_{\bs}(x)   - \tilde{\Delta}_{\bs}(x)   | < \frac{1}{2^m} \text{ and }
| \Delta_{\gs}(x,y) - \tilde{\Delta}_{\gs}(x,y) | < \frac{1}{2^m}
\end{equation}
in time\footnote{Note that $m$ here and elsewhere in the rest of paper is not related to the number of elements in the infrastructure.} $\text{poly}(n,m)$.
\item The minimum distance $d_{\min}$ between any two elements of the 
infrastructure is bounded from below by 
\ben
d_{\min} = \min_{x \in X} \left\{\Delta_{\bs}(x)\right\} \geq \frac{1}{2^{\text{poly}(n)}}.
\een
\item The maximum distance $d_{\max}$ between any two elements of the 
infrastructure is bounded from above by
\ben
d_{\max} = \max_{x\in X} \left\{\Delta_{\bs}(x) \right\} \leq {\text{poly}(n)}.
\een
\item There exists a positive integer $\bar{k} \le \text{poly}(n)$ and a positive (rational) number $ d_{\bar{k}} \ge \text{poly}(n)$ such that for all $x \in X$ we have 
\ben
\sum_{i=0}^{\bar{k}-1} \Delta_{\bs}(\bs^i (x)) \geq d_{\bar{k}}\,,
\een
where $\bs^i$ denotes the $i$-fold application of $\bs$.  In words, any $\bar{k}$ consecutive elements span a distance of a least $d_{\bar{k}}$.
\end{compactenum}

We emphasize that these assumptions are not restrictive; in fact, they are routinely made in the
work on infrastructures.  We have spelt them out explicitly for expository reasons.
In particular, infrastructures  arising from quadratic number fields satisfy all the 
assumptions made above; further justification for these assumptions for number fields is provided below.
The first three assumptions are obvious.  The relative distances $\Delta_{\bs}$ and $\Delta_{\gs}$ could be arbitrary real numbers and, thus, we cannot always
obtain the exact values.  Assumption \textbf{A}4 is made because  we
cannot perform arithmetic with arbitrary real numbers.
Assumptions \textbf{A}5 -- \textbf{A}7 ensure that we can compute in certain
circle groups associated to infrastructures and evaluate certain homomorphisms into these groups efficiently in time $\text{poly}(n)$.

The computational problems in infrastructures  are : 
\begin{itemize}
\item \textbf{Computation of the circumference}: \\
determine an $m$-bit approximation of the circumference $R$
\item \textbf{Generalized discrete logarithm problem}: \\given an element $y\in X$, determine an $m$-bit approximation of $d(y)$
\end{itemize}

The main contributions of this work are efficient quantum algorithms 
for infrastructures satisfying assumptions \textbf{A}1 -- \textbf{A}7. These algorithms make it possible to determine $\lfloor R \rceil$ and $\lfloor d(y) \rceil$ in time $\text{poly}(n)$, where the notation $\lfloor r \rceil$ means either the floor or ceiling of the real number $r$.  Simple classical post processing allows us to obtain efficiently $m$-bit approximations from these integral approximations. For the sake of completeness, we prove later how this can be accomplished.

We now justify the validity of the above assumptions in the case of infrastructures from number fields of unit rank $1$ (such number fields give rise to one-dimensional infrastructures). 
\begin{compactenum}[\textbf{A}1)]
\item This is shown in \cite{Sands91} (see also \cite{APK04}).
\item This is shown in \cite[Corollary 3.7]{thiel95}. 
\item In \cite{BW88}, it is shown that the baby steps and giant steps can be computed in $O(D^\epsilon)$ for arbitrary $\epsilon>0$ (where $D$ is the absolute value of the discriminate, which is bounded by $2^{\text{poly}(n)}$).  However, if one traces through their references and updates the analysis of the running time, one finds that everything
is polynomial in $\log(D)$ and not just subexponential \cite{FonteinJune11}. 
\item This assumption is valid since one can approximate logarithms of absolute values of elements in number fields whose size is polynomially bounded in $n$.
\item In \cite[Example 9.4]{Schoof08}, it is shown that $d_{\min}$ can be of size $1/2^{\text{poly}(n)}$.  In \cite{FonteinJune11}, Fontein informed us that \textbf{A}5 holds in general.
\item This is shown in \cite[Proposition 2.7 (i)]{BW88}.
\item This is shown in \cite[Proposition 2.7 (ii)]{BW88}.
\end{compactenum} 

The infrastructures from function fields are always discrete. This means that there are no issues with finite precision.  Therefore, the above computational problems can be solved directly with the standard hidden subgroup approach.  This is because the circle groups corresponding to discrete infrastructures are just finite cyclic groups.  In \cite{FonteinJune11}, Fontein informed us that the relevant assumptions also hold in infrastructures from finite fields.  

%
%

\subsection{Circle groups from infrastructures}

We now show that infrastructures naturally give rise to circle groups that are isomorphic to $\R/R\Z$.  This isomorphism is the key to solving the two computational problems in quantum polynomial time.  Here and in the next two subsections, we assume that we can compute $\Delta_{\bs}$ and $\Delta_{\gs}$ exactly.  

Picture the elements of $X$ to be embedded in a circle of circumference $R$ as follows.  They
are placed along the circle starting with $x_0$ at the topmost point of the circle and then moving clockwise.  Their position is determined by the distance function $d$. For instance, the element $x_i$ is associated to the point 
$d(x_i)$ on the circle as depicted in figure~\ref{fig:embedding}. 
\begin{figure}[h]
\begin{centering}
\begin{tikzpicture}
\draw (0,0) [thick]circle (2);
\foreach \x in {90, 80, 60, 30 , -30 , -110, -135, 150}
\draw (0,0) +(\x:1.9) -- (\x:2.1);
\draw (0,0) +(90:2) node[above] {$x_0$};
\draw (0,0) +(84:2) node[above right] {$x_1$};
\draw (0,0) +(-28:2) node[below right] {$x_i$};%
\end{tikzpicture}
\caption{Embedding an infrastructure into $\mbb{R}/R\mbb{Z}$} \label{fig:embedding}
\end{centering}
\end{figure}
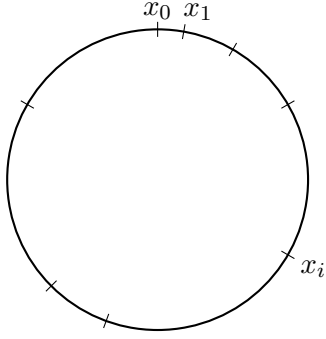

This embedding alone does not yet give rise to a valid group structure because $d(x_i)+d(x_j)$ is not necessarily an element of $d(X)$. 
To obtain a group, we start with the set $X\times\R$ and the map $\psi: X \times \mbb{R}\rightarrow \mbb{R}/R\mbb{Z}$ defined by
\ben
\psi(x,f) = d(x)+f\label{eq:absDist}
\een 
for all $(x,f)\in X\times\R$. We call this the absolute distance of the pair $(x,f)$.  

For each $d\in\R/R\Z$, there exist infinitely many pairs $(x,f)\in X\times\R$ with $\psi(x,f)=d$.  To avoid this infinitude, we continue by defining the equivalence relation $\equiv$ on $X\times\R$:  two pairs $(x,f),(y,g)\in X\times\R$ are said to be equivalent if and only if $\psi(x,f)=\psi(y,g)$ (which is the same as $d(x)+f \equiv d(y)+g \bmod R$).  We denote the equivalence class of $(x,f)$ by
$[x,f]$.

Now the set $X\times \mbb{R}/\equiv$ can be endowed with a group structure as follows.  

\medskip
\begin{prop} \label{prop:repGrp}
The absolute distance map $\psi$ in equation~(\ref{eq:absDist}) is a group isomorphism from $\cG:=X\times\R/\equiv$ to $\R/R\Z$, where the (commutative) group operation on $\cG$ is defined by
\begin{equation}
  [x,f] + [y,g] := [\gs(x,y), f + g - \Delta_{\gs}(x,y)]\label{eq:grpAdd}
\end{equation}
for arbitrary pairs $(x,f),(y,g)\in X\times\R$.
\end{prop}
\begin{proof}
The proof is straightforward. We just verify that $\psi$ is a group homomorphism.
Letting $\psi(x,f) = d(x)+f$ and $\psi(y,g) = d(y)+g$, 
we obtain $\psi(\gs(x,y),f+g-\Delta_{\gs}(x,y)) = d(\gs(x,y))+f+g-\Delta_{\gs}(x,y)$.
By the definition of the giant-step it holds that $d(\gs(x,y)) = d(x)+d(y)+\Delta_{\gs}(x,y)$. Thus,
$d(\gs(x,y))+f+g-\Delta_{\gs}(x,y)= d(x)+d(y)+f+g = \psi(x,f)+\psi(y,g)$.
\end{proof}

%
%

\subsection{Group arithmetic based on $f$-representations}

We have to use ``nice'' representatives for the equivalence classes of $\cG$ to be able to compute within this group efficiently.  To this end, we introduce $f$-representations.  Intuitively, the $f$-representations fill in the missing points in the circle $\mbb{R}/R\mbb{Z}$, i.e., the set of points $(\mbb{R}/R\mbb{Z})\setminus d(X)$. 

\medskip
\begin{defn}[$f$-representation]
Let $(X,d)$ be an infrastructure.  A pair $(x,f) \in X\times \mbb{R}$ is said to be an
$f$-representation if $0\leq f< \Delta_{\bs}(x)$.  We denote the set of all $f$-representations by $\rep(\mc{I})$.
\end{defn}

The following lemma was shown in \cite{fontein08} (see Proposition~2 and Corollary~1 therein) in a slightly less general setting.  We include this lemma for completeness.  An important aspect of this lemma is that the group operation can be realized without having any knowledge of $R$ or the distance function $d$ (except for the knowledge that is is revealed indirectly through the particular interplay of functions $\bs$, $\gs$, $\Delta_{\bs}$, and $\Delta_{\gs}$).  

We mention that for arbitrary infrastructures, neither  this lemma nor any simple method make it possible to compute inverses in $\cG$.  However, in the case of infrastructures in global fields there is an efficient classical way to compute (approximate) $f$-representations of inverses in the corresponding circle groups. 

%
%

\medskip
\begin{lemma}\label{lm:grpAdd}
The group operation in $\cG$ can be efficiently realized by using $f$-representations to encode the equivalence classes. More precisely, it takes at most $\bar{k} \ceil{2d_{\max}/d_{\bar{k}}}=\text{poly}(n)$ invocations of baby steps to obtain the $f$-representation corresponding to the sum of two elements of $\cG$.
\end{lemma}
\begin{proof}
Let $(x,f), (x',f')\in\rep(\mc{I})$. Then, we have
\be
[x,f] + [x',f'] = [\gs(x,x'), f + f' - \Delta_{\gs}(x,x')] \,.
\ee
In general, the pair $(x'',f''):=(\gs(x,x'), f + f' - \Delta_{\gs}(x,x'))\in X\times\R$ is not a valid $f$-representation.  The task  now is to find the $f$-representation that encodes the same 
equivalence class in $\cG$ as $(x'',f'')$.  We use the bounds  
$-d_{\max}\leq f'' = f+f'-\Delta_{\gs}(x,x') \leq f+f' \leq 2d_{\max}$, where $d_{\max}$ is 
the  maximum distance between two consecutive elements of the infrastructure. 

If $f''\leq 0$, then we iteratively replace $(x'',f'')$ with $(\bs^{-1}(x''),f''+\Delta_{\bs}(x''))$ until it just becomes positive.  If 
$f''\geq 0$, then we iteratively replace $(x'',f'')$ with $(\bs(x''), f''-\Delta_{\bs}(x''))$ until it is minimal while being nonnegative. 
Observe that this reduction process preserves the absolute distance.  Moreover, it takes at most 
$\bar{k} \ceil{2d_{\max}/d_{\bar{k}}}=\text{poly}(n)$ steps to obtain to the canonical representative in $\rep(\mc{I})$. 
\end{proof}

From now on, we identify $\cG$ with $\rep(\mc{I})$ and use $(x,f)\in\rep(\mc{I})$ to denote the group elements instead of $[x,f]$ to simplify notation.

The corollary below is a simple consequence of the above lemma.  We state it explicitly 
because this result it is extensively used in
the quantum algorithms. 

%
%
\medskip
\begin{corollary}[Double \& multiply]\label{co:sqrMult}
Let $(x,f)\in\cG$ be an arbitrary group element and $a\in \mbb{Z}$ an arbitrary nonnegative integer.  Then, it takes at most 
$O(\bar{k} \ceil{2d_{\max}/d_{\bar{k}}} \log(a)) = \text{poly}(n) \log(a)$ invocations of baby steps and at most $O(\log(a))$ invocations of giant steps to obtain the $f$-representation 
corresponding to $a \cdot (x,f)$.
\end{corollary}
\begin{proof}
The action of $\Z$ on the commutative group $\cG$ is defined by
\be
a\cdot (x,f) & := & \underbrace{(x,f) + (x,f) + \cdots + (x,f)}_{a \text{  times}}\,.\label{eq:multRep}
\ee 
Consider the special case of computing $a\cdot (x,f)$ for $a=2^i$ with some  $i$. This takes at most $O(i)$ steps: 
\ben
(x,f)       & = & (x^{(0)},f^{(0)}) \nonumber\\
2 (x,f)     & = & (x^{(0)},f^{(0)}) + (x^{(0)},f^{(0)}) = \big(\gs(x^{(0)},x^{(0)}), 2f^{(0)} - \Delta_{\gs}(x^{(0)},x^{(0)})\big) \nonumber\\
&=:& (x^{(1)},f^{(1)}) \nonumber\\
            & \vdots &  \nonumber\\
2^{i} (x,f) & = & (x^{(i-1)},f^{(i-1)}) + (x^{(i-1)},f^{(i-1)})  \nonumber\\
            & = & \big(\gs(x^{(i-1)},x^{(i-1)}), 2 f^{(i-1)} - \Delta_{\gs}(x^{(i-1)},x^{(i-1)})\big) \label{eq:sqrMult}\\
            & =:& (x^{(i)},f^{(i)}) \nonumber\,. 
\een
In each step, we apply the above lemma to ensure that $(x^{(i)},f^{(i)})$ are valid $f$-representations.  

Now suppose $a = b_i 2^{i} + b_{i-1} 2^{i-1} + \cdots + b_0 2^{0}$ in binary representation.  Then, $a \cdot (x,f)$ can be computed as 
\be
a\cdot (x,f) = \sum_{j=0}^i b_j \cdot (x^{(j)},f^{(j)})\,.
\ee
with at most $i$ additions.  We again use the above lemma to ensure that the partial sums are valid $f$-representations.

In total, the whole process takes at most $O(\log (a))$ giant-steps and $O(\bar{k} \ceil{2d_{\max}/d_{\bar{k}}}\log (a))$  baby-steps.
\end{proof}

%
%

\subsection{Group homomorphisms from $\R$ and $\Z \times \R$ into circle groups}

In this subsection, we continue to assume that we can determine
the functions $\Delta_{\bs}$ and $\Delta_{\gs}$ exactly, and compute with arbitrary real numbers. In the next subsection, we will relax this assumption.

\begin{defn}
Let $h : \R \rightarrow \cG$ be the surjective group homomorphism, where $h(r)$ is defined to be the unique $f$-representation $(x,f)\in\rep(\mc{I})$ with $(x_0,r)\equiv (x,f)$.  
\end{defn}

Recall that we define the distance function $d$ such that $d(x_0)=0$, thus $(x_0,0)$ is the
identity of $\mc{G}$.

The statement of following lemma is obvious.  We formulate it explicitly since it provides the intuition required to understand the quantum algorithm for computing the circumference.
\medskip

\begin{lemma}
The kernel of $h$ is equal to $R\Z$.  Thus, $h$ is a periodic function on $\R$ with period $R$.
\end{lemma}

\medskip
\begin{lemma}\label{lm:discJumps}
Let $r\in [0,B]\subset\R$, where $B$ is an arbitrary (but fixed) positive real number.  Then, we can determine the exact value
$h(r)$  
using $O(\log(B))$ giant-steps and $O(\log(B)\bar{k}\ceil{2d_{\max}/d_{\bar{k}}})=O(\log(B)\text{poly}(n))$ baby-steps
under the assumption that $\Delta_{\bs}$ and $\Delta_{\gs}$ can be computed exactly.
\end{lemma}
\begin{proof}
In general, $(x_0,r)$ is not a valid $f$-representation.  Thus, we need to find the corresponding $f$-representation. If $r$ is small and positive, then we can use baby-steps to find it with at most $\bar{k}\ceil{r/d_{\bar{k}}}$ invocations. 

If $r$ is large, then the baby-step method is not efficient anymore. We have to use giant-steps as well.  The idea is to use the
double and multiply technique of Corollary~\ref{co:sqrMult}. Let $x_{\bar{k}}=\bs^{\bar{k}}(x_0)$. Then $d(x_{\bar{k}}) \geq d_{\bar{k}}$.  Let $a=[ r/d(x_{\bar{k}})]$, where 
$[x ]$ denotes the nearest integer to $x $. We can compute
$a\cdot (x_{\bar{k}},0) =(x,f)$ using $O(\log (a))= O(\log(B))$ giant-steps
and $O(\log (B) \bar{k} \ceil{2d_{\max}/d_{\bar{k}}})$ baby-steps. Note that $(x,f) \equiv (x_0, ad({x_{\bar{k}}})) $. But, $|ad({x_{\bar{k}}}) - r| = |[ r/d(x_{\bar{k}})]d({x_{\bar{k}}}) - r|\leq d(x_{\bar{k}})/2$.
Therefore, $(x,f)$ is at most within a distance of $d(x_{\bar{k}})/2$ from $r$. Thus we can find $h(r)$ by using no more than $\bar{k}$ additional invocations of either $\bs$ or $\bs^{-1}$. 
The overall time complexity of evaluating $h(r)$ is therefore $O(\log(B)\bar{k}\ceil{2d_{\max}/d_{\bar{k}}})= O(\log(B)\text{poly}(n))$, since $d_{\max}$ and $\bar{k}$  are $O( \text{poly}(n)$ by assumptions {\bf A}6 and {\bf A}7.
\end{proof}
Similar ideas can be applied when $r$ is negative.
The method proposed in Lemma~\ref{lm:discJumps} relies essentially on the group arithmetic of 
$\mc{G}$ and thus is quite different from a generalization of the binary search method.

\medskip
\begin{defn}
Let $x\in X$ be an arbitrary (but fixed) element of the infrastructure.  Let $g : \Z \times \R \rightarrow \cG$ be the surjective homomorphism, where $g(a,r)$ is defined to be the unique $f$-representation corresponding to
\ben
a \cdot (x,0) + h(r)\,.\label{eq:g}
\een
\end{defn}
We note that  $g(a,b)$ is same as the $f$-representation of $h(ad(x)+b)$,  where $d(x)$ is the distance of $x$.

The following statement on the kernel of the homomorphism $g$ is obvious.  

\medskip
\begin{lemma}\label{lm:ker_g}
The kernel of the above homomorphism $g$ is equal to
\be
\{ (a, r) \, : \, r \equiv - a \, d(x) \mod R \}\,.
\ee
\end{lemma}

\medskip
\begin{corollary}
Let $A$ be an arbitrary positive integer and $B$ an arbitrary positive real number.  Then, we can determine the exact value $g(a,b)$ for all pairs $(a,r)\in\{0,1,\ldots,A-1\}\times [0,B]$ in time $O((\log A+ \log B)\text{poly}(n))$ under the assumption that $\Delta_{\bs}$ and $\Delta_{\gs}$ can be computed perfectly.
\end{corollary}
\begin{proof}
By definition  $g(a,r)= a\cdot(x,0)+h(r)$. The computation of $a\cdot (x,0)$ can be 
performed in  $O(\log (A)\bar{k} \ceil{2d_{\max}/d_{\bar{k}}}) =O(\log(A)\text{poly}(n))$ time by Corollary~\ref{co:sqrMult}, while the computation of $h(r)$ can be performed in 
$O(\log (B)\bar{k}\ceil{2d_{\max}/d_{\bar{k}}}) = O(\log(B)\text{poly}(n))$ time by Lemma~\ref{lm:discJumps}. The final group addition in 
$\mc{G}$ takes at most $\bar{k}=\text{poly}(n)$ baby-steps, by Lemma~\ref{lm:grpAdd}.
\end{proof}

%
%

\subsection{Efficient approximate group arithmetic and evaluation of the homomorphisms from $\R$ and $\Z\times\R$}\label{ssec:effArithmetic}

The previous assumption that we can compute $\Delta_{\bs}$ and $\Delta_{\gs}$ and represent arbitrary real numbers is clearly an idealization.  We made this assumption at first because we can explain the intuition in a simpler and more elegant way when the homomorphisms $h$ and $g$ are perfect.  We now drop this assumption and work instead with the approximate versions $\tilde{\Delta}_{\bs}$ and $\tilde{\Delta}_{\gs}$. 

Let $L$ be some large positive integer.  We only consider evaluation points $r$ that are rational numbers with denominator $L$.  

Let $h(r)=(x,f)$ be the perfect $f$-representation with $(x,f)\equiv (x_0,r)$.  We can only determine an approximate $\tilde{h}(r)=(\tilde{x},\tilde{f})\in X\times\R$ of $h(r)$.  This approximation can be realized efficiently and has the  following two properties:
\begin{enumerate}[\bf {P}1.]
\item The first component is  off at most by either a baby-step backward or forward, i.e., $\tilde{x} \in \{\bs^{-1}(x), x, \bs(x)\}$.
\item If we have the promise that 
\begin{equation}\label{eq:farAway}
\frac{1}{L} \leq f \leq \Delta_{\bs}(x)-\frac{1}{L}
\end{equation} 
holds, then the first component is correct, i.e., $\tilde{x}=x$, and the second component $\tilde{f}$ satisfies 
\begin{equation}
|f-\tilde{f}|\le\frac{1}{2L} \,.\label{eq:hPrecision}
\end{equation}
\end{enumerate}

Later, we will show that all evaluation points $r$ necessary for the quantum algorithm are such that the condition in equation~(\ref{eq:farAway}) holds with high probability by adding 
a random shift to the evaluation points.

\medskip
\begin{lemma}[Approximate homomorphism $\tilde{h}$]\label{lm:nonDiscJumps}
Let $L$ be a positive integer with $d_{\min}>1/L$.  We consider only evaluation points of the form $r=k/L$ with $r<B$. 
Let $h(r)=(x,f)$ be the perfect $f$-representation.  Then, we can compute an approximate pair $\tilde{h}(r)=(\tilde{x},\tilde{f})$ that satisfies {\bf P1}, {\bf P2}.  The running time is $\text{poly}(\log(B),\log(L),n)$.
\end{lemma}
\begin{proof}
We analyze what happens if we run the algorithm in Lemma~\ref{lm:discJumps}, but now rely on the approximate versions $\tilde{\Delta}_{\bs}$ and $\tilde{\Delta}_{\gs}$.  
Recall that the parameter $m$ characterizes the precision of the approximations.  The maximal deviation between the approximate and perfect values is smaller than $1/2^m$.  

We use $\tilde{d}_{\acc}(\cdot)$ to denote the corresponding approximate accumulated distances 
 of the (intermediate) $f$-representations and their first components.
We use $d_{\acc}(\cdot)$ to denote the correct accumulated distance of the representations and elements (these distances exist even though we cannot always compute them). The accumulated 
distances are not taken modulo $R$ and take into account how the $f$-representation is generated. 
A key observation that we need in the proof is that 
$d_{\acc}(\tilde{x},\tilde{f}) = d_{\acc}(\tilde{x})+\tilde{f}$ and 
$\tilde{d}_{\acc}(\tilde{x},\tilde{f}) = \tilde{d}_{\acc}(\tilde{x})+\tilde{f}$, so that
$d_{\acc}(\tilde{x},\tilde{f})- \tilde{d}_{\acc}(\tilde{x},\tilde{f})  = 
d_{\acc}(\tilde{x})- \tilde{d}_{\acc}(\tilde{x})$.

The characterizing condition of the perfect $f$-representation is 
\begin{equation}\label{eq:perfectX}
d_{\acc}(x) \le r < d_{\acc}(x) + \Delta_{\bs}(x)\,.
\end{equation}
We can only guarantee
\begin{equation}\label{eq:ApproxAcc}
\tilde{d}_{\acc}(\tilde{x}) \le r < \tilde{d}_{\acc}(\tilde{x}) + \tilde{\Delta}_{\bs}(\tilde{x})
\end{equation}
for the approximate pair $(\tilde{x},\tilde{f})$.

Assume that $m$ has been chosen to be sufficiently large so that 
\begin{equation}
| \tilde{d}_{\acc}(\tilde{x}) - d_{\acc}(\tilde{x}) | \le \frac{1}{2L}\label{eq:accDistCondition}
\end{equation} 
holds.  Together with equation~(\ref{eq:ApproxAcc}) this implies
\begin{equation}
d_{\acc}(\tilde{x}) - \frac{1}{2L} \le r < d_{\acc}(\tilde{x}) + \Delta_{\bs}(\tilde{x}) + \frac{1}{2L} + \frac{1}{2^m}\,.
\end{equation}
This condition on $\tilde{x}$ is weaker than the condition of the perfect $x$ in equation~(\ref{eq:perfectX}).  But since $1/2^m < 1/L < d_{\min}$ we must have
$\tilde{x}\in\{\bs^{-1}(x),x,\bs{(x)}\}$, depending on which of the three cases $r<d_{\acc}(\tilde{x})$, $d_{\acc}(\tilde{x})\le r < \tilde{d}_{\acc}(\tilde{x}) + \Delta_{\bs}(\tilde{x})$, or 
$\tilde{d}_{\acc}(\tilde{x}) + \Delta_{\bs}(\tilde{x}) \le r$ occurs.  We cannot have a deviation by more than one baby-step backward or forward because otherwise equation~(\ref{eq:ApproxAcc}) would
not be satisfied.

If we know that $f$ satisfies $\frac{1}{L} \leq f \leq \Delta_{\bs}(x)-\frac{1}{L}$, then we can conclude that $\tilde{x}=x$ must hold.  This is because the first and third cases are excluded.  The condition on $\tilde{f}$ is automatically satisfied in this case since
$\tilde{f} = r - \tilde{d}_{\acc}(\tilde{x})$, which is the same as $r - \tilde{d}_{\acc}(x)$.

We now show how to choose $m$ so that the condition in equation~(\ref{eq:ApproxAcc}) holds.
The algorithm in Lemma~\ref{lm:discJumps} has two steps. In the first step, we
compute $a\cdot (x_{\bar{k}},0)$, where  $a=[r/d(x_{\bar{k}})]$. This gives us a representation $(x',f')$, such that 
\be
|d_{\acc}(x',f')-r|\leq \frac{d(x_{\bar{k}})}{2}.
\ee
Then we apply a sequence of baby-steps to obtain an $f$-representation $(x,f)$, 
which satisfies $d_{\acc}(x,f) = r$.

Working with $\tilde{\Delta}_{\bs}$ and $\tilde{\Delta}_{\gs} $, in the first step we actually
compute $(\tilde{x}',\tilde{f}')$ an approximation of $\tilde{a} \cdot (x_{\bar{k}},0) $, where 
$\tilde{a}=[r/\tilde{d}({x_{\bar{k}}})]$. 
 
Let us analyze the error in this computation. The computation of $\tilde{a} \cdot (x_{\bar{k}},0)$ itself can be broken down into two parts: (i) computation of representations of the form $(\tilde{x}^{(i)}, \tilde{f}^{(i)})$ which approximate $2^i(x_{\bar{k}},0)$ and (ii) summing  $O(\log \tilde{a})$ such representations.

The error at the very beginning $e_0$ satisfies
\begin{equation}
e_0 := | \tilde{d}_{\acc}(\tilde{x}^{(0)}) - d_{\acc}(\tilde{x}^{(0)}) |= | \tilde{d}_{\acc}(\tilde{x}^{(0)},\tilde{f}^{(0)}) - d_{\acc}(\tilde{x}^{(0)},\tilde{f}^{(0)}) | < \frac{\bar{k}}{2^m}\,.\nonumber
\end{equation}
Note that $\tilde{d}_{\acc}(\tilde{x}^{(0)}) \ge d_{\bar{k}}$ holds because if we get a value strictly smaller than $d_{\bar{k}}$ we can replace it by $d_{\bar{k}}$, because of {\bf A}7.  The error in the 
$i$th step 
\begin{equation}
  e_i:=| \tilde{d}_{\acc}(\tilde{x}^{(i)}) - d_{\acc}(\tilde{x}^{(i)}) | = 
  | \tilde{d}_{\acc}(\tilde{x}^{(i)}, \tilde{f}^{(i)}) - d_{\acc}(\tilde{x}^{(i)}, \tilde{f}^{(i)}) |\nonumber
\end{equation}
satisfies the recursion 
\begin{equation}
e_i < 2 e_{i-1} + \frac{1}{2^m} + \frac{\bar{k}}{2^m}\ceil{\frac{2d_{\max}}{d_{\bar{k}}}}
\end{equation}
The recursion relation can be easily explained by considering equation~\eqref{eq:sqrMult}. 
The first term is due to the fact that the error in $\tilde{f}^{(i-1)}$ is
multiplied by 2,  the second term is due to one
giant-step, and the third term due to $O(\bar{k}\ceil{2d_{\max}/d_{\bar{k}}})$ baby-steps used to obtain a valid $f$-representation.
This implies
\begin{equation}
e_i < \frac{1}{2^{m-i}}\left(\bar{k}+1+\bar{k}\ceil{\frac{2d_{\max}}{d_{\bar{k}}}} \right)\,.
\end{equation}
In order to obtain $(\tilde{x}',\tilde{f}')$, we have to sum $O(\log \tilde{a})$ such 
$f$-representations, where $i $ varies from $0$ to $\log \tilde{a} -1$. 
Each sum adds an additional error term due to the giant step and the baby-steps used for
reduction. Therefore the error at the end of the first step is given by 
\be
e'&: =&| \tilde{d}_{\acc}(\tilde{x}') - d_{\acc}(\tilde{x}') | = 
  | \tilde{d}_{\acc}(\tilde{x}', \tilde{f}') - d_{\acc}(\tilde{x}', \tilde{f}') |\\
  &=& \frac{\tilde{a}}{2^{m}}\left(\bar{k}+1+\bar{k}\ceil{\frac{2d_{\max}}{d_{\bar{k}}}} \right) + \frac{\log \tilde{a}}{2^m} \left( 1+ \bar{k}\ceil{\frac{2d_{\max}}{d_{\bar{k}}}}\right)\,,\\
  &\leq& \frac{r+d_{\bar{k}}}{2^{m}d_{\bar{k}}}\left(\bar{k}+1+\bar{k}\ceil{\frac{2d_{\max}}{d_{\bar{k}}}} \right) + \frac{\log (r/d_{\bar{k}}+1)}{2^m}\left( 1+ \bar{k}\ceil{\frac{2d_{\max}}{d_{\bar{k}}}}\right)\nonumber \,
\ee
where we used the fact that $\tilde{a} \leq r/d_{\bar{k}}+1$.
The $f$-representation $(x',f')$ is at most at a distance\footnote{We can tighten this by a factor of 2. But this suffices. } of $d_{\max}\bar{k}$ from $r$. 
Thus $(\tilde{x}',\tilde{f}')$ is at most a distance of $(e'+{d_{\max}}\bar{k})$ from $r$
and we need to take at most $\bar{k}\ceil{(e'+{d_{\max}}\bar{k})/d_{\bar{k}}}$ baby-steps to obtain 
$(\tilde{x},\tilde{f})$.

 The error in the accumulated distances of the final representation $(\tilde{x},\tilde{f})$ is given by 
 \be
 \tilde{e}&: =&| \tilde{d}_{\acc}(\tilde{x}) - d_{\acc}(\tilde{x}) | = 
  | \tilde{d}_{\acc}(\tilde{x}, \tilde{f}) - d_{\acc}(\tilde{x}, \tilde{f}) |\\
  &=& e'+ \frac{\bar{k}}{2^m}\ceil{\frac{e'+d_{\max}\bar{k}}{d_{\bar{k}}}}
 \ee
The dominant term in the error is the first term $e'$, as it is proportional to $r$, while the
second term is proportional to $r/2^m$ and therefore does not contribute too much as $m$ is  
large. We can make the error smaller than $1/2L$ as required in equation~\eqref{eq:accDistCondition} by choosing $m=\text{poly}(\log(B), \log(L))$. 
\end{proof}
The proof does not actually require that the evaluation points are of the form $k/L$.

Analogous results hold for the homomorphism $g$.  We state them without proof since the above argument can be easily  adapted.

Let $g(a,r)=(x,f)$ be the perfect $f$-representation with $(x,f)\equiv (x_0,r)$.  We can only determine an approximate $\tilde{g}(a,r)=(\tilde{x},\tilde{f})\in X\times\R$ of $g(a,r)$.  This approximation can be realized efficiently and has the  properties {\bf P1}, {\bf P2}.

\medskip

\begin{lemma}[Approximate homomorphism $\tilde{g}$]\label{lm:approxg}
Let $L$ be a positive integer with $d_{\min}>1/L$.  We consider only evaluation points of the form $(a,r)$ with $a\in\{0,1,\ldots,A-1\}$ and $r=k/L\in [0,B]$.
Let $g(a,r)=(x,f)$ be the perfect $f$-representation.  Then, we can compute an approximate pair $\tilde{g}(a,r)=(\tilde{x},\tilde{f})$ that satisfies {\bf P1}, and {\bf P2}. The running time is $\text{poly}(\log(A),\log(B),\log(L),n)$.
\end{lemma}

%
%

\section{Quantum algorithm for approximating the period of pseudo periodic states}
\label{sec:periodicStates}
In this section we generalize the notion of periodic states introduced 
in \cite{KLM07}. We assume that the quantum states are elements of a $q$-dimensional complex Hilbert space, denoted by $\C^q$.

%
%

\subsection{Pseudo-periodic states}

\begin{defn}[Periodic state]
A quantum state  in   $\C^q$  is  periodic with period 
$r\in \mbb{Z}$ at offset $k\in\{0,1,\ldots, r-1 \}$ if it is of the form
\ben
\ket{\psi}_{k,r} :=\frac{1}{\sqrt{p}}\sum_{j=0}^{p-1}\ket{k+jr},\label{eq:periodicState}
\een
where $p = \floor{ (q-k-1)/r+1}$. We denote a periodic state with period $r$ at offset $k$ by
$\ket{\psi}_{k,r}$.
\end{defn}
Periodic states can be created by the evaluation of injective functions over a uniform 
superposition. To be more precise, we create the state 
$\ket{\psi}= q^{-1/2}\sum_{i=0}^{q-1}\ket{i}\ket{f(i)}$,
and measure the second register. We assume that $f$ is periodic with period $r$.
It is possible to recover the period $r$ by means of Fourier sampling. In fact, the period 
can be recovered even when $r$ is irrational. For this reason, we generalize these periodic 
states to a larger class of quantum states called the pseudo-periodic states. 

\medskip

\begin{defn}[Pseudo-periodic state]
A pseudo-periodic state in $\C^q$, with possibly irrational period $r\in \mbb{R} $, is of the form:
\ben
\ket{\psi}_{k,r} &= &\frac{1}{\sqrt{p}} \sum_{j=0}^{p-1} \ket{\lfloor k+j r\rceil},\label{eq:irrPeriodicState}
\een
where $k\in \{0,1,\ldots, \floor{r} \}$ and  $p$ is the largest integer such that 
$\lfloor k+(p-1)r\rceil\leq (q-1)$. 
\end{defn}
Please note that $\lfloor x \rceil $ can be either $\floor{x}$ or $\ceil{x}$, therefore,  
 $p$ take any integer value in the set
$\{\floor{(q-2)/r},\ldots, \floor{q/r}+1\}$, depending on the value of the offset $k$. If we assume that $r>2$, then we can restrict 
$p \in \{\floor{q/r}-1,\floor{q/r}, \floor{q/r}+1 \}$.

The weakly periodic functions defined in \cite{hallgren07} are one class of functions which 
can induce such pseudo-periodic states.
As we  show in this section, we can  recover the period even when the state is ``almost''
periodic.
 We observe
that in the definition of the periodic states above, there is an implicit
dependence on the offset $k$; this offset is usually the outcome of some measurement, and
therefore random. 

%
%

\subsection{Perturbed geometric sums with missing terms}

The following lemma is  at the heart of the analysis of the quantum algorithms for infrastructures.  It is   crucial for  understanding the performance of these algorithms. 
The special  case $\mc{J}=\{0,1,\ldots, n-1 \}$ suffices to bound the probability of the 
algorithm for computing the circumference. The more general case where $\mc{J}$ is a proper
subset of $\{0,1,\ldots, n-1 \}$ is necessary for the analysis of the quantum algorithm for
computing the discrete logarithms. 

\medskip

\begin{lemma}[Perturbed geometric sums with missing terms]\label{lm:pertMissingSeries}
Let  $\omega$ be the $n$th root of unity $e^{2\pi i/n}$, $n\geq 2$, $\theta$ an arbitrary real-valued function defined on $\mc{J}\subseteq \{0,\ldots,n-1\}$ satisfying the following conditions on 
$\theta_j$ and $|\mc{J}|$:
\begin{subequations}
  \begin{equation}
|\theta_j| \le n/32
    \label{eq:thetaDist}
  \end{equation}
  \begin{equation}
|\mc{J}| \geq n(1-c_\delta)/(1-2\sin(\pi/32))
    \label{eq:JSize}
  \end{equation}
\end{subequations}
where 
\ben
c_\delta &=&  \rm{sinc} (\delta) = \frac{\sin(\pi \delta)}{\pi \delta} \quad  \text{ if } |\delta| < 1.\label{eq:cdelta}
\een
Then the following inequality holds:
\begin{eqnarray}
\frac{1}{|\mc{J}|^2}
\left|
\sum_{j\in \mc{J}}
\omega^{\delta j + \theta_j}
\right|^2 
& \ge & 
\left(
1 
- 2 \sin(\pi/32)
- (
1-c_\delta
) \frac{n}{|\mc{J}|}\right)^2 \,.
\end{eqnarray}
\end{lemma}
\begin{proof}
Triangle inequality and upper bound on the absolute value of the unperturbed geometric sum without missing terms imply
\begin{eqnarray}
\left|
\sum_{j\in \mc{J}}
\omega^{\delta j} 
\right|
+
\left|
\sum_{j\in \bar{\mc{J}}}
\omega^{\delta j}
\right| 
& \ge &
\left|
\sum_{j\in \mc{J}}
\omega^{\delta j} 
+
\sum_{j\in \bar{\mc{J}}}
\omega^{\delta j}
\right| \nonumber \nonumber\\
& = &
\left|
\sum_{j=0}^{n-1}
\omega^{\delta j}
\right| =
\left|
\frac{1 - \omega^{\delta n}}{1 - \omega^{\delta}} 
\right| \nonumber\\
& = &
\left|
\frac{\sin( \pi \delta)}{\sin(\pi\delta/n)} 
\right|\label{eq:sinEq}\\
& \geq & \left|\frac{\sin (\pi \delta)}{\pi\delta/n}\right| = n c_\delta,\label{eq:sinEq2}
\end{eqnarray}
The  equality in equation~\eqref{eq:sinEq} follows from $|1-e^{i \vartheta}| = |e^{-i\vartheta/2} - e^{i \vartheta/2}| = 2|\sin(\vartheta/2)|$ holding for all $\vartheta\in\mbb{R}$.  For the inequality in equation~\eqref{eq:sinEq2}  we used the fact that $|\sin \vartheta| \leq |\vartheta|$, when $|\vartheta|< \pi/2$.

Subtracting the absolute value of the sum over $\bar{\mc{J}}$ from both sides of 
equation~\eqref{eq:sinEq2}
and dividing by $|\mc{J}|$ yields
\begin{eqnarray}
\frac{1}{|\mc{J}|}
\left|
\sum_{j\in \mc{J}}
\omega^{\delta j} 
\right| 
& \ge &
c_\delta \frac{n}{|\mc{J}|} - 
\frac{1}{|\mc{J}|} \left| \sum_{j\in \bar{\mc{J}}}
\omega^{\delta j}
\right| \nonumber\\ 
& \ge &
c_\delta\frac{n}{|\mc{J}|} - 
\frac{\bar{\mc{J}}}{|\mc{J}|}\nonumber \\
& = &
1 - \Big(
1-c_\delta\Big) \frac{n}{|\mc{J}|}\,.\label{eq:unpertBound}
\end{eqnarray}

We now bound the ``perturbed'' geometric sum.  To this end, we use some basic ideas from quantum information theory. Define the states $|\psi\>=\frac{1}{\sqrt{|\mc{J}|}} \sum_{j\in\mc{J}}  \omega^{\delta j} |j\>$ and $|e\> = \frac{1}{\sqrt{|\mc{J}|}} \sum_{j\in \mc{J}}|j\>$, the projector $P=|e\>\<e|$, and the diagonal unitary matrix $U=\mathrm{diag}(\omega^{ \theta_0},\ldots, \omega^{\theta_{|\mc{J}|-1}})$.  Observe that the square of the absolute value of unperturbed geometric sum is equal to $\| P |\psi\> \|^2$ and that of the perturbed one to $\| P U |\psi\> \|^2$. We have
\begin{eqnarray}
\Big| \| P |\psi\> \| -  \| P U |\psi\> \| \Big| 
& \le &
\| P |\psi\> -  P U |\psi\> \|  \nonumber\\
& \le & 
\| P \| \cdot \| I - U \| \cdot \| |\psi\> \| \nonumber \\
& = &
2 \max_j \Big\{ |\sin(2 \pi \theta_j/ (2 n))| \Big\} \nonumber\\
& \le &
2 \sin(\pi/32)\,.\label{eq:inter2}
\end{eqnarray}
The upper bound on $\|I-U\|$ follows by noting that the entries of the diagonal matrix $I-U$ are $1-e^{2\pi i \theta_j/ n}$ and using the above identity for the absolute value of expressions of this form. Let $\left \| P |\psi\> \right\| =x$
and $\| P U |\psi\> \|=y$. Then equation~\eqref{eq:inter2}  
implies the desired result since 
\begin{equation}
y^2 \ge \big(x-2\sin(\pi/32) \big)^2 \ge \left(1- 2 \sin(\pi/32)- \left(1-
c_\delta\right)\frac{n}{|\mc{J}|} 
\right)^2\,
\end{equation}
where we used equation~\eqref{eq:unpertBound} in the last step.
\end{proof}

We pause to make two observations regarding the application of this result. First, 
we must ensure that $|\mc{J}|/n \geq (1-c_{\delta})/(1-2\sin(\pi/32)) $ for $\delta \in [0, 1)$.
 Second, the choice of $|\theta_j|\leq n/32$,  can be improved  in that we can tolerate a higher perturbation, depending on the actual value of 
$\delta$. Although, we retain this bound on $\theta_j$ throughout this paper for the
sake of a clearer exposition, 
optimizing this bound on $\theta_j$ based on $\delta$ will enable us to obtain better bounds 
on the success probability of the quantum algorithms.

%
%

\subsection{Presentation and proof of the quantum algorithm}

Now we shall give a quantum algorithm for estimating the period of a pseudo-periodic
state. In general, these states arise  from some periodic functions, therefore the proposed quantum algorithm 
can be used to estimate the periods of such functions. 

\begin{theorem}\label{th:qalgoPeriodicState}
Given a pair of  pseudo-periodic states   whose period $S\in \mbb{R}$  is bounded as 
 $M\geq S>1$, then with a probability $\Omega(1)$ 
and in time $\text{poly}(\log S)$,
Algorithm~\ref{alg:periodicStates} gives a list of real numbers $\mc{L}$ such that for some $\hat{S}\in \mc{L}$, we have $|S-\hat{S}| \leq 1$. Further, 
$|\mc{L}| = O ( \text{poly}\log S )$ and the success probability 
is given by 
\ben
p_{\rm{success}}\geq \frac{1 }{2}\left({\frac{1}{32}-\frac{2}{S}}\right)^2
\left({1-\frac{2S}{q}}\right)^2 \left( \sinc\Big(\frac{1}{2}+\frac{1}{2S}\Big)- 2 \sin(\pi/32) 
\right)^4\,\, \label{eq:psuccessPeriodic}
\een
where $M^2\leq q<2M^2$. 
\end{theorem}
\begin{algorithm}
\caption{{\ensuremath{\mbox{\scshape Approximate period of pseudo-periodic states}}}
}\label{alg:periodicStates}
\begin{algorithmic}[1]
\REQUIRE {A pair of pseudo-periodic states in $\C^q$ with period $S\in \mbb{R}$, where  
$M$ is an upper bound  on $ S>2$ and $q$ is an integer such  that  $S^2\leq M^2\leq q< 2M^2$.}
\STATE For each pseudo-periodic state,  apply a Fourier transform over $\mbb{Z}_q$ and measure  to obtain $c$ and $d$.
\STATE Compute the convergents $c_i/d_i$ of $c/d$ where $d_i\leq  \floor{q/32}$.
\STATE Return $\mc{L}= \Big\{ [c_iq/c] \mid d_i\leq \floor{q/32} \Big\}$ as candidates for $S$.
\end{algorithmic}
\end{algorithm}

\begin{proof}
Assume that  the pseudo-periodic state is as follows:
\begin{equation}
\ket{\psi}_{o,S} = \frac{1}{\sqrt{|\mc{J}|}}
\sum_{j\in \mc{J}} | \lfloor o + j S\rceil \>\,.\nonumber
\end{equation}
where $  \mc{J} = \{0,1,\ldots, p-1 \}$ and $p \in \{ \floor{q/S}-1, \floor{q/S},\floor{q/S}+1 \}$.
Since we are Fourier sampling, we may assume without loss of generality, that $o=0$.
Therefore, the measured distribution will be the same as the one induced by Fourier sampling the following state:
\begin{equation}
\frac{1}{\sqrt{|\mc{J}| }}
\sum_{j\in \mc{J}} \ket{ \lfloor j S\rceil}
\end{equation}
Taking the Fourier transform over $\mbb{Z}_q$ we obtain
\ben
\frac{1}{\sqrt{|\mc{J}| }}\frac{1}{\sqrt{q}}
\sum_{j\in \mc{J}} \sum_{\ell=0}^{q-1}\omega_q^{\ell \lfloor  j S \rceil}\ket{\ell}.
\een
The Fourier transform at $|\ell\>$ has the amplitude
\begin{equation}
\frac{1}{\sqrt{q|\mc{J}| }} \sum_{j \in \mc{J}} \omega_q^{\lfloor  jS\rceil\ell}\,.
\end{equation}
We seek to find a lower bound on the probability of obtaining outcomes $\ell$ of the form $[ \frac{mq}{S} ] $, where $m\in\{0,1,\ldots,\lfloor S \rfloor\}$.  For a given $m$, $[ m \frac{q}{S} ]$ denotes either the floor or ceiling so that that $m \frac{q}{S} = [ m \frac{q}{S}] + \epsilon_\ell$ with $|\epsilon_\ell|\le \frac{1}{2}$.
The probability of observing $\ell$ is given by
\begin{equation}
\frac{1}{q|\mc{J}| } \left| \sum_{j\in\mc{J}}  \omega_p^{\frac{p}{q}\lfloor jS\rceil [m\frac{q}{S}]} \right|^2\,.
\end{equation}
To bound this probability, we consider the exponent of $\omega_p$
\begin{eqnarray}
\frac{p}{q}\lfloor j S\rceil [ m\frac{q}{S} ]
& = &
\frac{p}{q}( j S + \delta_j)(m\frac{q}{S} + \epsilon_\ell) \nonumber \\
& = &
p m j  + \frac{p S \epsilon_\ell}{q} j  + \frac{p\delta_j \epsilon_\ell}{q} + \frac{pm\delta_j}{S}\, \nonumber\\
& = & p  m  j+  \frac{p S \epsilon_\ell}{q} j  + \frac{p\delta_j \epsilon_\ell}{q} + \frac{pm\delta_j}{S}\, \nonumber\\
&= & pm j + \delta j + \theta_j.
\end{eqnarray}

The first term is a multiple of $p$, implying that it can be omitted in the exponent.  The factor $\delta=\frac{p S \epsilon_\ell}{q}$ in front of $j$ in the second term is less or equal to $(1+S/q)/2 \leq (1+1/S)/2$. 
The absolute value of the sum of the third and fourth terms is less or equal to $p/32$ provided that $m< \lfloor S/32\rfloor$.  In this case, the phase perturbations $\theta_j$ caused by these two terms satisfy equation~\eqref{eq:thetaDist}. Further, $|\mc{J}|=p$ ensures that equation~\eqref{eq:JSize} is also satisfied and we can apply Lemma~\ref{lm:pertMissingSeries}.
We conclude that the probability of obtaining $\ket{\ell}$ is 
\be
\frac{1}{q|\mc{J}|  } \left| \sum_{j\in\mc{J}}  \omega_p^{\frac{p}{q}\lfloor j S\rceil [m\frac{q}{S}]} \right|^2 &\geq& \frac{p  }{q} \left(
 \sinc\Big(\frac{1}{2}+\frac{1}{2S}\Big) 
- 2 \sin(\pi/32)
\right)^2\,\\
&\geq& \left(\frac{1 }{S}-\frac{2}{q}\right) \left(
\sinc\Big(\frac{1}{2}+\frac{1}{2S}\Big)- 2 \sin(\pi/32) 
\right)^2\,
\ee
where the last inequality follows from  $p\geq\floor{q/S} -1\geq q/S-2$. So the probability of obtaining any ``good'' $\ell$, i.e. $m\in \{1,\ldots,\floor{S/32-1}\}$, is at least $\beta$, where 
\ben
\beta = \left(\frac{S}{32}-2\right) \left(\frac{1}{S}-\frac{2}{q}\right)\left(
 \sinc\Big(\frac{1}{2}+\frac{1}{2S}\Big) - 2 \sin(\pi/32)
\right)^2\,,\label{eq:goodPairProb}
\een
where we used that $\floor{S/32-1}\geq (S/32-2)$.
The measured value $\ell$ is a multiple of $q/S$ rounded
to the nearest integer i.e. $\ell = [ mq/S ]$ for some $m$.

Unlike the case of period finding algorithm where the period is integral, 
the period $S$ of $\ket{\psi}_{o,S}$ cannot be reconstructed
with Fourier sampling one (pseudo-periodic) quantum state. However, as shown below, 
we can reconstruct using the method suggested by Hallgren in \cite{hallgren07}.
 Suppose we have two  measurements 
$c= [ kq/S]$ and $d=[ l q/S ]$,  obtained
by Fourier sampling the pair of periodic states, 
then $k/l$ occurs as a convergent of $c/d$ and
we can compute an integer close to 
$S$  by computing $[kq/c]$. 
Without loss of generality assume that
$0<k\leq l < \floor{S/32}$. Assume that $c=kq/S+\epsilon_c$  and $d=lq/S+\epsilon_d$ where
$-1/2\leq\epsilon_c,\epsilon_d\leq 1/2$.
Then 
\be
\left|\frac{c}{d}-\frac{k}{l} \right| &=& \left|\frac{kq+\epsilon_c S}{lq+\epsilon_d S} - \frac{k}{l}\right|
=  \left|\frac{S(\epsilon_c l-\epsilon_dk)}{l^2q+\epsilon_dS l} \right|\\
&\leq& \left|\frac{S( l+k)/2}{l^2q -S l/2} \right| \leq  \left|\frac{Sl}{l^2q-Sl/2}\right| \\
&=&   \left|\frac{1}{lq/S-1/2}\right|  < \frac{1}{2l^2},
\ee
under the assumption that $0<k\leq l< \floor{S/32}$ and $q\geq S^2$. Thus ${k}/{l}$ is a 
convergent of ${c}/{d}$. 
Since $l\leq \floor{S/32}$, we only need to
compute the convergents $c_i/d_i$ whose denominators $ d_i$ are less than $\floor{q/32}$. 
We now form the list of candidate estimates for $S$ as 
\ben
\mc{L}= \Big\{ [c_iq/c] \mid d_i\leq \floor{q/32}\Big\}.\label{eq:candSolns}
\een
As the $d_i$ grow exponentially, $|\mc{L}| =O(\text{polylog}(|S|))$.

Since $k/l$ is a convergent of $c/d$, we know that there exists an estimate 
 $\hat{S}=[{kq}/{c} ] \in \mc{L}$. We now show that $\hat{S}$ 
satisfies $|S-\hat{S}|\leq1$. Let $c=kq/S+\epsilon_c$ and $\hat{S} = kq/c $, where $|\epsilon_c |  \leq1/2$.  Then,  we can bound 
$|S-\hat{S}|$ as 
\be
|S-\hat{S}|& = &\left|S - \frac{S}{1+\epsilon_c S/kq} \right| \leq \left|\frac{\epsilon_c S^2/kq}{1+\epsilon_c S/kq}\right| \\
&\leq &\left|\frac{\epsilon_c/k}{1+\epsilon_c/kS}\right| \mbox{ because } q\geq S^2\\  
&=&  \left| \frac{\epsilon_c}{k}\right| \left|\frac{1}{1+\epsilon_c/kS}\right| \\
&\leq& \frac{1}{2k} \frac{1}{1-1/2kS}  \leq \frac{1}{2k}\cdot 2\\
&\leq&1.
\ee 
We now compute a lower bound on the success probability of the algorithm. 
We have already seen that the probability of  a  pair of good measurements is given by 
\eqref{eq:goodPairProb}. In order to be 
able to recover the period $S$, we require  $k$ and $l$ to be coprime. By 
Lemma~\ref{lm:coprimeProb}, the probability that $k$, $l$ are coprime is at least $1/2$.
Thus the overall success probability of 
the algorithm is greater than ${\beta^2 /2}=\Omega(1)$.  
\end{proof}

The algorithm does not return a single value for the period but
rather a small list of candidates for the period. This presumes a post processing step 
by which  we can single out the solutions. 

Further, we note that the previous algorithm uses a pair of pseudo-periodic states and if these 
states are being prepared probabilistically, then we must factor that into the success 
probability of the algorithm.

%
%

\section{Quantum algorithm for approximating the circumference}\label{sec:circumference}

Our goal  is to set up  pseudo-periodic states 
whose period is a multiple of the circumference of an infrastructure.
Then the quantum algorithm 
of the  preceding section can be applied to extract an integer close to the circumference. 
With this knowledge,  the circumference can be computed to the desired
accuracy by a classical algorithm.

\subsection{Pseudo-periodic states from infrastructures}

In section~\ref{ssec:effArithmetic}, we showed
that an approximate version $\tilde{h}$ of $h$    can be computed so that  properties {\bf P1}, {\bf P2} are satisfied. For this approximate version to be useful, it is necessary 
that the $f$-representations at the evaluation points meet the condition stated in 
 equation~\eqref{eq:farAway}. In 
this subsection, we show how to satisfy this condition  which allows us to compute
$\tilde{h}$ so that the first component is always correct and the error in the second 
component is under control. However, $\tilde{h}$ does not induce the periodic states that
we discussed in the previous section.  To create a periodic quantum state it is essential to work with a ``quantized'' version of 
$\tilde{h}$. Therefore we  introduce the function $h_N :\mbb{Z}\rightarrow X\times  \mbb{Z}$ 
by setting 
\ben
h_N(i)=(\tilde{x},\lfloor \tilde{f} N \rfloor),\label{eq:hN}
\een
where $\tilde{h}(\frac{i}{N}+\frac{j}{L})=(\tilde{x},\tilde{f})$. When {\bf P2} is satisfied,
it is helpful to interpret $h_N$ in the following way:  $h_N(i)= ({x},k )$, then $k$ is the number of sampling points between $d({x})+\lfloor (i/N + 
j/L)/R\rfloor R$ and $i/N+j/L$. 

The incorrectness in $\tilde{h}$ cannot be avoided if the evaluation points $r$
are chosen arbitrarily. As already stated in Lemma~\ref{lm:nonDiscJumps}, we assume that the evaluation
points are of the form $k/L$ for some large integer $L$ and bounded $k$. Even so, we cannot evaluate always 
$\tilde{h}$ correctly for every $k$. Therefore, we further restrict the evaluation of $\tilde{h}$ to a subset of
the points which are $\frac{1}{N}$ uniformly spaced along a bounded interval, where
$N $ divides $L$.  We choose $N\ge \lceil{2}/{d_{\min}}\rceil$
so that there are at least two evaluation points $\frac{i}{N}$ and $\frac{i+1}{N}$ between any two adjacent elements of $\mc{I}$.
This is shown in the figure below. The dashed lines indicate the sampling points.

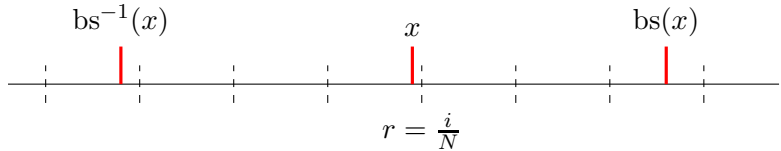
\begin{figure}[h]
\begin{center}
\begin{tikzpicture}[scale=2.5]
\clip (1.3,-0.4) rectangle (5.5,0.5);
\draw (-0.4,0) -- (5.4,0);
\foreach \x in {0.5, 1, 1.5, 2, 2.5 ,3, 3.5,4, 4.5,5}
\draw[dashed] (\x,-0.1) -- (\x,0.1);
\foreach \x in {0.25, 1.9,3.45,  4.8}
\draw[color=red, very thick] (\x,0) -- (\x,0.2);
\draw (1.9,0.2) node[above] {${\bs}^{-1}(x)$};
\draw (3.45,0.2) node[above] {$x$};
\draw (4.8,0.2) node[above] {$\bs(x)$};
\draw (3.5,-0.1) node[below] {$r=\frac{i}{N}$};
\end{tikzpicture}
\end{center}
\end{figure}

But this is still inadequate to satisfy equation~\eqref{eq:farAway}, as some of the evaluation 
points could be very close to elements of the infrastructure. So we 
shift all the evaluation points by a random offset of the form $\frac{j}{L}$, where 
$j$ is chosen uniformly at random from $\{0,1,\ldots,\frac{L}{N}-1\}$.  This is shown in the figure 
below.  The solid lines indicate the shifted evaluation points and they are still of the form 
$k/L$.

\begin{figure}[h]
\begin{center}
\begin{tikzpicture}[scale=2.5]
\clip (1.3,-0.4) rectangle (5.5,0.5);
\draw (-0.4,0) -- (5.4,0);
\foreach \x in {0.5, 1, 1.5, 2, 2.5 ,3, 3.5,4, 4.5,5}
{
\draw[dashed] (\x,-0.1) -- (\x,0.1);
\draw (\x+0.2,-0.1) -- (\x+0.2,0.1);
}
\foreach \x in {0.25, 1.9,3.45,  4.8}
\draw[color=red, very thick] (\x,0) -- (\x,0.2);
\draw (1.9,0.2) node[above] {${\bs}^{-1}(x)$};
\draw (3.45,0.2) node[above] {$x$};
\draw (4.8,0.2) node[above] {$\bs(x)$};
\draw (3.5+0.2,-0.1) node[below] {$r=\frac{i}{L}+\frac{j}{L}$};
\end{tikzpicture}
\end{center}
\end{figure}
Now we can show that with high probability equation~\eqref{eq:farAway}
is satisfied and can use Lemma~\ref{lm:nonDiscJumps} to guarantee that $\tilde{h}$ can be 
computed with the  precision stated in equation~\eqref{eq:hPrecision}.

\medskip
\begin{lemma}\label{lm:offset}
Let  $N\geq \ceil{2/d_{\min}}$. Suppose we  evaluate $h_N$ at points $i/N+j/L$ for $i\in\{0,\ldots,q-1\}$, where
 $j$ chosen uniformly at random from $\{0,1,\ldots,L/N-1\}$, and $L$ is an integer such that
 \ben
 L\geq N\ceil{ \frac{2\bar{k} }{(1-p_{h})} \ceil{\frac{q }{N d_{\bar{k}}} }}. \label{eq:offset}
 \een
 Then  with probability greater or 
 equal to $p_{h}$, no sampling point $i/N+j/L$ is closer than $1/L$ to any element $x$ of the 
 $\mc{I}$, i.e., 
 \ben
 |(d(x)-i/N-j/L) \bmod R | \geq 1/L.
 \een 
\end{lemma}
\begin{proof}
By assumption \textbf{A}7 there are at most $\bar{k}\ceil{q/Nd_{\bar{k}}}$ elements of $\mc{I}$ in the interval $[0,q/N]$.
There are $L/N$ possible offsets to choose from.  Since the offsets are spaced at $1/L$,
any element $x\in \mc{I}$ can be within a distance of less than $1/2L$ for at most two offsets. 
The fraction of offsets that are not useful is given by $2\bar{k}\ceil{q/Nd_{\bar{k}}}/(L/N) \leq 1-p_{h}$ provided that $L$ is chosen as  in equation~\eqref{eq:offset}.
\end{proof}

When $L$ is chosen according to Lemma~\ref{lm:offset}, we have  $h_N(i)=(x,\lfloor \tilde{f} N \rfloor)$,
where $\tilde{h}(\frac{i}{N}+\frac{j}{L})=(x,\tilde{f})$.  We use $x$ instead of $\tilde{x}$ on purpose to emphasize again that the first component is correct.  It is crucial to observe that $\lfloor \tilde{f} N\rfloor$ is equal to $\lfloor f N\rfloor$.  This is because {\bf P2}
holds and no evaluation point is within $1/L$  of any element of the infrastructure.

The preceding results imply that $h_N(i)$ can be computed efficiently and correctly. 

\medskip

\begin{corollary}If Lemma~\ref{lm:offset} holds, then for 
 all $i$ with $0\leq i\leq 2N^2R^2$ the value $h_N(i) = (\tilde{x},\lfloor \tilde{f} N \rfloor)$ is equal to $(x,\floor{fN})$, where $\tilde{h}(i/N+j/L)=(\tilde{x},\tilde{f})$ and
$h(i/N+j/L)=(x,f)$.  
\end{corollary}

Next we show that $h_N$ when evaluated over a  finite interval induces a periodic state with 
probability greater than or equal to $ 1/2$, if we assume that no sampling point is too close to any element of the infrastructure. 
\medskip

\begin{lemma}\label{lm:periodicState}
Let $N\geq \ceil{ 2/d_{\min} }$ and let $\ket{\psi} = q^{-1/2}\sum_{i=0}^{q-1} \ket{i}\ket{h_N(i)}$. We assume that no
element of the infrastructure is too close to the sampling points $i/N+j/L$, where $j$ and $L$ are chosen as in Lemma~\ref{lm:offset}.
Then, with probability greater than 
\ben
p_{\text{periodic}}= \left(1-\frac{1}{Nd_{\min}}-\frac{1}{NR}\right) \left(1-\frac{2NR}{q} \right)\label{eq:periodicProb}
\een
measuring the second register of 
$\ket{\psi}$ induces a periodic state with period $NR$,
\ben
\ket{\psi}_{k,NR} = \frac{1}{\sqrt{p}} \sum_{\ell=0}^{p-1} \ket{\lfloor k+\ell NR\rceil},
\een 
where $p$ is equal to one of the values\footnote{Note that $N\geq \ceil{2/d_{\min}}$, implies that 
$NR>2$, and therefore, $p$ must be at least $\floor{q/NR}-1$.} $\floor{q/NR} - 1$, $\floor{q/NR}$, or $\floor{q/NR} +1$. 
\end{lemma}
\begin{proof} 
Denote the measurement outcome by $(x,m)$.  First, we show that if $(x,m)$
satisfies a certain condition, then the resulting post-measurement state is a pseudo-periodic state.  Second, we estimate the probability that we obtain such measurement
outcome.

Assume that $h_N(k)=(x,m)$ for some $k\in\{0,\ldots,\floor{NR}\}$.  
Then, in $\ell$th period the sampling points are at a distance $\alpha_\ell+m_{\ell}/N$ for  $m_\ell\in \{0,1, \ldots, \floor{N\bs(x)} \}$ from the element $x$.  Under the assumption of Lemma~\ref{lm:offset}, $1/L\leq \alpha_\ell \leq 1/N-1/L$. 

Consider now the sampling points for the zeroth period and some other period $\ell\neq 0$. 
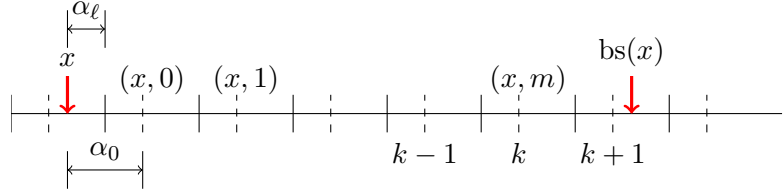
\begin{figure}[h]
\begin{center}
\begin{tikzpicture}[scale=2.5]
\clip (1.3,-0.4) rectangle (5.5,0.6);
\draw (-0.4,0) -- (5.4,0) ;
\foreach \x in {0.5, 1, 1.5, 2, 2.5 ,3, 3.5,4, 4.5,5}
{
\draw[dashed] (\x,-0.1) -- (\x,0.1);
\draw (\x-0.2,-0.1) -- (\x-0.2,0.1);
}
\foreach \x in {0.25, 1.6, 4.6}
\draw[color=red, very thick,<-] (\x,0) -- (\x,0.2);
\draw (1.6,0.2) node[above] {$x$};
\draw (4.6,0.2) node[above] {$\bs(x)$};
\draw (3.5,-0.1) node[below] {$k-1$};
\draw (4,-0.1) node[below] {$k$};
\draw (4.5,-0.1) node[below] {$k+1$};
\draw (2.05,0.05) node[above] {$(x,0)$};
\draw (2.55,0.05) node[above] {$(x,1)$};
\draw (4.05,0.05) node[above] {$(x,m)$};
\draw (1.6,-0.2)--(1.6,-0.4);
\draw (2,-0.2)--(2,-0.4);
\draw[<->] (1.6,-0.3)--(2,-0.3);\draw (1.8,-0.3) node[above] {$\alpha_0$};
\draw (1.8,0.35)--(1.8,0.55);
\draw (1.6,0.35)--(1.6,0.55);
\draw[<->] (1.6,0.45)--(1.8,0.45);\draw (1.7,0.45) node[above] {$\alpha_\ell$};
\end{tikzpicture}
\end{center}
\end{figure}

Then, the following cases arise: $1/L\leq \alpha_\ell \leq \alpha_0 $,  and 
$\alpha_0< \alpha_\ell \leq 1/N-1/L$. As can be seen from the figure above,
if $1/L\leq\alpha_\ell \leq \alpha_0$, then  we must have $h_N(k)=h_N(k+\floor{\ell NR}) $. On the other hand, 
if $\alpha_0 < \alpha_\ell \leq 1/N-1/L $, then it is clear that $h_N(k)=h_N(k+\ceil{\ell NR})$ unless $k$ 
corresponds to the last sampling point between the elements $x$ and $y=\bs(x)$ since in this case 
$h_N(\ceil{k + \ell N R}) = (y,0) \neq h_N(k)$. 

On the one hand, if $k$ does not correspond to the last sampling point between two adjacent 
elements of $\mc{I}$, then for all $\ell \in \{0,1,\ldots, p-1 \}$ we have $h_N(k+\lfloor \ell N R\rceil)=h_N(k)$.
On the other hand, if $k$ corresponds to the last evaluation point between two elements, then 
the preimage may not contain all $\ell$.
 
We now estimate the probability of obtaining an outcome $(x,m)$ such that $h_N(k)=(x,m)$ and the offset $k\in\{0,\ldots,\floor{NR}\}$ 
does not correspond to the last evaluation point between any two elements. 

There are $\floor{NR}+1$ possible offsets in the zeroth period.  At most $\ceil{R/d_{\min}}$ of these can correspond to last evaluation points between two elements. 
We know that the preimage of a ``good'' measurement outcome $(x,m)$ contains at least $\floor{q/NR}-1$ elements.  So, the probability of obtaining a good measurement
outcome is at least
\be
p_{\text{periodic}}&=&\frac{(\floor{NR}+1 - \ceil{R/d_{\min}}) \cdot (\floor{q/NR}-1)}{q}\\
&\geq& (NR-R/d_{\min}-1) (q/NR-2)/q\\
& = &\left(1-\frac{1}{Nd_{\min}}-\frac{1}{NR}\right) \left(1-\frac{2NR}{q} \right).
\ee
\end{proof}

\subsection{Presentation and proof of the quantum algorithm}

\begin{theorem}[Estimating the circumference to arbitrary accuracy]\label{lm:accuracy}
Let $\mc{I}$ be an infrastructure satisfying the assumptions \textbf{A}1--\textbf{A}7.  For 
any $\delta > 0$, there is an efficient Las Vegas algorithm that outputs an estimate $\hat{R}$ of the circumference $R$ of $\mc{I}$ such that $|R-\hat{R}| \le \delta$.  

Let $N\geq \ceil{2/d_{\min}}$, $S=NR$,  $p_h$  the probability of evaluating $h_N$ correctly, and 
$p_{\text{periodic}}$ the probability of creating a periodic state,  see equation~\eqref{eq:periodicProb}. Then,
the classical algorithm  invokes Algorithm 1 an expected $O(1/q_{\mathrm{success}})$ number of times, 
where $q_{\rm{success}}$ is
\ben
q_{\rm{success}}\geq \frac{ p_h^2 p_{\text{periodic}}^2}{2}\left({\frac{1}{32}-\frac{2}{S}}\right)^2
\left({1-\frac{2S}{q}}\right)^2 \left(
\sinc\Big(\frac{1}{2}+\frac{1}{2S}\Big) - 2 \sin(\pi/32)\right)^4.\,\,
\een
The classical computations take $\mathrm{poly}(\log(R), \log(1/\delta))$ time.
\end{theorem}
%
%
\begin{proof}  We first create an 
pseudo-periodic state in $\C^q$, where $q$ is chosen as specified by 
Algorithm~\ref{alg:periodicStates}.
We create the superposition  
\[
\ket{\psi} = \frac{1}{\sqrt{q}}\sum_{i=0}^{q-1}\ket{i}\ket{h_N(i)}.
\]
If the conditions of Lemma~\ref{lm:offset}  are satisfied, then $\ket{\psi}$  will be created correctly with a probability $p_h$.  Then by 
Lemma~\ref{lm:periodicState}, measuring the second register of the state results in a  periodic state $\ket{\psi}_{k,NR}$ with probability $\geq 1/2$, where $p \in \{\floor{q/S}-1,\floor{q/S},\floor{q/S}+1 \}$. 
Algorithm~\ref{alg:periodicStates} returns $\mc{L}$, a list of candidates for $S$, which
contains an element $\hat{S}$ which satisfies $|S-\hat{S}|\leq 1$.
The probability of this event is  
\begin{equation}
\Pr(
|S-\hat{S}|\le 1) \ge  p_h^2p_{\text{periodic}}^2p_{\mathrm{success}}\,,
\end{equation}
where $p_{\mathrm{success}}$ is defined in equation~\eqref{eq:psuccessPeriodic}.
 The factor of $p_h^2p_{\text{periodic}}^2$ is due to the fact that the Algorithm~\ref{alg:periodicStates} needs to create a pair of the pseudo-periodic states. 

Assume that $|S-\hat{S}|\le 1$ is present (of course, we do not know this).  This is equivalent to $|R-R'|\le 1/N$, where $R'=\hat{S}/N$.  We actually check for a slightly weaker
condition namely, $|(R-R') \bmod R| \leq 1/N$. But this suffices. 

Recall that we always choose $N\ge \lceil {2}/{d_{\min}} \rceil$.  This implies that either $h(R')=(x_0,f)$ with $f\le\frac{1}{N}$ or $h({R'})=(\bs^{-1}(x_0),g)$ with $g\ge\Delta_{\bs}(\bs^{-1}(x_0))-1/N$.  If we evaluate $\tilde{h}$, the approximate version  of $h$, at $R'$ with precision  $\delta_{\mathrm{prec}}\le\frac{1}{2N}$, then it remains the case that we can only obtain either $(x_0,\tilde{f})$ or $(\bs^{-1}(x_0),\tilde{g})$. If so we can conclude that $|R-R' \bmod R | \leq 1/N$. 

Now assume that $|(R-R') \bmod R| > 1/N$ holds.  In this case, we may or may not encounter $\bs^{-1}(x_0)$ or $x_0$ by evaluating $\tilde{h}$ at $R'$.

Because our test actually checked for $ |(R-R')\bmod R| \leq 1/N $, we could have some
spurious solutions when $R'$ is a multiple of $R$.  If this
is the case, then we return the smallest such $R'$ as satisfying $ |R-R'  | \leq 1/N$. We
then obtain an estimate for $R$ as follows. 

Once we have encountered $\bs^{-1}(x_0)$ or $x_0$, we can compute $\tilde{h}(R')$ with precision $\delta/2$.  If we obtain $(\bs^{-1}(x_0),\tilde{g})$, then we set 
\begin{equation}
\hat{R} = R' - \tilde{g} + \Delta_{\bs}(\bs^{-1}(x_0))\,,
\end{equation}
where we compute the distance $\Delta_{\bs}$ with precision $\delta/2$.  If we obtain $(x_0,\tilde{f})$, then we set $\hat{R}=R'-\tilde{f}$.  All these computations can be carried out in $\mathrm{poly}(\log(R), \log(1/\delta)$ time.

The expected number of times we have to invoke the quantum algorithm to encounter $\bs^{-1}(x_0)$ or $x_0$ is clearly at most $1/q_{\mathrm{success}}$.
\end{proof}

There is a subtle point  worth spelling out. In each run of the algorithm, 
there are two evaluations of $h_N$. We assume that the same random shift is used in both 
these evaluations and in any subsequent $O(1/q_{\rm{success}})$ runs. Only if the algorithm
fails in all these runs do we change the offset and repeat the process.

Finally, it can be easily verified for sufficiently large $S$, say $S\geq 256$, the 
lower bound on the success probability is greater than a constant, irrespective of the
size of the problem.

The proposed algorithm when specialized to number fields improves upon \cite{hallgren07} in the 
following aspects.  
The probability of success of the proposed algorithm is lower
bounded by equation~\eqref{eq:psuccessPeriodic} which is a constant $10^{-5}$ as opposed to \cite{hallgren07} for which the success probability decreases as $\Omega(1/\log ^4 (M))$, where $M\geq NR$ and is lower than $10^{-9}$, see \cite[Claim~3.5~and~Lemma~3.4]{hallgren07} therein. As the expression indicates the success probability of the algorithm decreases with increasing circumference and the performance gap with respect to our algorithm increases. Our lower bound is better than
the lower bound of \cite{schmidt09}, namely $2^{-26}$.
Our result also implies fewer repetitions to boost the probability of success, thereby lowering the
complexity of the algorithm. In addition, the proposed algorithm requires a smaller Quantum 
Fourier transform, thereby lowering the number of qubits and circuit complexity.

%
%
%
%

\section{Quantum algorithm for solving the generalized discrete logarithm problem in infrastructures}\label{sec:dlog}

In this section we give a quantum algorithm for the discrete logarithm problem. 
Given an element $x$ of an infrastructure $\mc{I}=(X,d) $ we
are required to find the distance of $x$, namely $d(x)$. 

The function  that is of interest in the computation of the discrete log problem is given by $g(a,b):\mbb{Z}\times \mbb{R}\rightarrow \mc{I}\times \mbb{R}$ where $g(a,r)= a\cdot (x,0) + h(r)$. 
By Lemma~\ref{lm:approxg} we can compute $\tilde{g}$ the approximate version of $g$, 
so that it satisfies properties {\bf P1}, and {\bf P2}.

As in the circumference case, we evaluate $\tilde{g}$ at carefully selected points to ensure that the first component is always correct and quantize the second component. This resulting
function is 
$g_N(a,b) : \mbb{Z}\times \mbb{Z}\rightarrow \mc{I}\times \mbb{Z}$ 
\ben
 g_N(a,b)= \left(\tilde{y},\floor{\tilde{f}N}\right),
\een
where $\tilde{g}(a,b/N+j/L)=(\tilde{y},\tilde{f})$. 

The first component of $g_N$ is correct provided that equation~\eqref{eq:farAway} is satisfied for all evaluation points of $g_N$, i.e.,   
none of the evaluation points are closer than $1/L$ to any element of the 
infrastructure.  As in the case of $h_N$, we achieve this with  high probability by
applying a random shift of the form $j/L$. The following lemma shows how to find a suitable $L$.

%
%
\begin{lemma}[Offset for DLOG]\label{lm:dlogOffset}
Suppose $\mc{I}$ is an infrastructure that satisfies the assumptions \textbf{A}1-7.
Let $\cA\subseteq \{0,1,\ldots,A-1\}$ and $\cB\subseteq\{0,1,\ldots,\floor{RN}-1\}$. Let 
\ben
L \geq \ceil{\frac{2 A \bar{k}}{(1-p_g)}  \ceil{\frac{1}{d_{\bar{k}}}\left(R-\frac{1}{N}\right)} }N. \label{eq:dlogL}
\een 
Let $j\in \{0,1, \ldots ,L/N-1\}$ be chosen uniformly at random.
Then, the probability that 
\begin{equation}
\left|
(a d_x + \frac{b}{N} + \frac{j}{L} - d_y) \bmod R
\right| \ge \frac{1}{L}
\end{equation}
holds for all $(a,b)\in\cA\times\cB$ and all $y\in X$ is greater or equal to $p_g$.
\end{lemma}
%
%
\begin{proof}
Consider a fixed $a\in\cA$, then all the points $ad_x+b/N+j/L$ are contained in the interval 
$\left[a d_x + j/L, a d_x + (\floor{R N-1})/N + j/L\right]$. This interval contains at most $\bar{k} \ceil{(R-1/N) / d_{\bar{k}}}$ elements $y\in X$
since its length is $\floor{RN-1}/N\le(R-1/N)$.  Observe that no $y\in X$ can be closer than $1/L$ to any evaluation point of the above form for more than two offsets.

Hence, if we consider all $a\in\cA$, then at most $2A \bar{k} \ceil{(R-1/N)/d_{\bar{k}}}$ offsets are bad.
Assuming $L$ as stated above, this implies that the probability that there is at least one element and at least one evaluation
point that are closer than $1/L$ to each other is at most  $(2A \bar{k} (R-1/N) /d_{\bar{k}})/(L/N) \leq 1-p_g$.
\end{proof}

We always compute $\hat{R}$ with sufficiently high precision so that $|\hat{R}-R|<1/(2N)$ holds.  
Then, we have $\hat{R}>R-1/2N$ and a suitable choice for $L$ would be 
$\ceil{2 A \bar{k}\ceil{\hat{R}/d_{\bar{k}}}/(1-p_g)}N$.

In the quantum algorithm for approximating the circumference, we encounter superpositions of the form:
\be
\ket{\psi} =\frac{1}{\sqrt{|\cA_{x,m}|}}\sum_{a\in\cA_{x,m}} \ket{a}\ket{(x,m)},
\ee 
where $\cA_{x,m}$ has the special form $\{ \fc{k + j R N} \, : \, j=0,\ldots,p\}$ and $(x,m)$ is equal to $h_N(k)$.

A somewhat similar type of quantum state appears in the discrete logarithm problem.  A major difference is that it involves a function of two variables
\be
\ket{\psi} =\frac{1}{\sqrt{|{\cA_{ y,\ell}}|}}\sum_{(a,b)\in \mc{J}} \ket{a}\ket{b} \ket{(y,\ell)},
\ee 
where $\cA_{y,\ell}$ is now the fiber over $(y,\ell) \in \im g_N$, i.e., $g_N(a,b)=(y,\ell)$
for $(a,b)\in \cA_{y,\ell}$. 

The intuition based on Lemma~\ref{lm:ker_g}, which characterizes the kernel of the perfect function $g$, suggests that the elements in 
$\cA_{y,\ell}$ lie ``close'' to a line whose slope encodes the distance of the element $x$.  This statement is proved in 
Lemma~\ref{lm:g_NKer}, which establishes the exact relation between $a$ and $b$ for $g_N$.  Lemma~\ref{lm:fib_gN_size},
establishes upper and lower bounds on the size of the preimage of $(y,\ell)$. 

The intuition based on the quantum algorithm for the discrete logarithm problem in finite cyclic groups suggests that we can extract the slope by Fourier sampling.  This statement is proved
in Theorem~\ref{th:dlog}.

\begin{lemma}\label{lm:g_NKer}
Let $\emptyset\neq\cA\subseteq\{0,1,\ldots,A-1\}$ where $A$ is a positive integer and $\cB\subseteq\{0,1,\ldots,\floor{RN}-1\}$.  Denote by $g_N(\cA\times\cB)$ the image of the function $g_N$, i.e.,
\begin{equation}
g_N(\cA\times\cB) = \{ g_N(a,b) \,:\, a\in\cA\,,\, b\in\cB \}\,.
\end{equation}
For each $(y,\ell)\in g_N(\cA\times\cB)$, the preimage $g_N^{-1}(y,\ell)$ has the form
\begin{equation}
g_N^{-1}(y,\ell) = \{ (a,b_a) \, : \, a\in \cA_{y,\ell}\}\,,
\end{equation}
where $\cA_{y,\ell}\subseteq \cA$ and 
assuming that a random shift of $j/L$ has been applied to the evaluation points, 
the values $b_a$ satisfy the condition
\begin{equation}
\left\lfloor \frac{a d_x+\frac{b_a}{N}+\frac{j}{L}}{R} \right\rfloor R +
d_y + \gamma_a + \frac{\ell}{N} =
a d_x + \frac{b_a}{N} + \frac{j}{L}
\end{equation}
with $1/L \le \gamma_a \le 1/N-1/L$.  The 
cardinality of the image satisfies the inequalities
\begin{equation}
|\cB| \le |g_N(\cA\times\cB)| \le \floor{R(N+1/d_{\min})}\,.\label{eq:dlogPreImage}
\end{equation}
\end{lemma}
%
%
\begin{proof}
Let $(y,\ell)\in g_N(\cA\times\cB)$ be arbitrary. Suppose that $(a,b_a) \in g_N^{-1}(y,\ell)$.
Then we must have 
\ben
d_y+\frac{\ell}{N}+\gamma_a &\equiv& ad_x+ \frac{b_a}{N}+\frac{j}{L} \bmod R \nonumber\\
&=&ad_x+ \frac{b_a}{N}+\frac{j}{L} - \floor{\frac{ad_x+ \frac{b_a}{N}+\frac{j}{L}}{R}}R,\nonumber
\een
where $1/L\leq \gamma_a \leq 1/N-1/L$. This constraint on $\gamma_a$ is due to the fact that
none of the sampling points are within a distance of less than $1/L$ from the elements of the infrastructure. 

The second component $\ell$ is bounded from above by 
\begin{equation}
\ell \le N\Delta_{\bs}(y) \nonumber
\end{equation}
since the inequality
\begin{equation}
\left\lfloor \frac{a d_x+\frac{b_a}{N}+\frac{j}{L}}{R} \right\rfloor R +
d_y + \gamma_a + \frac{\ell}{N} <
\left\lfloor \frac{a d_x+\frac{b_a}{N}+\frac{j}{L}}{R} \right\rfloor R +
d_y + \Delta_{\bs}(y) \nonumber
\end{equation}
holds for all $(a,b_a)$ with $g_N(a,b_a)=(y,\ell)$.  This implies that the number of images whose first component is equal to $y$ is at most $N\Delta_{\bs}(y)+1$.  Summing over all elements of the infrastructure yields the upper bound $RN+R/d_{\min}$.  We can improve this to $\floor{R(N+1/d_{\min})}$ since the cardinality of $g_N(\cA\times \cB)$ is an integer.  Hence, $|g_N(\cA\times\cB)|\le\floor{R(N+1/d_{\min})}$.
\end{proof}

A condition similar to equation~\eqref{eq:dlogPreImage} has been established in 
\cite{hallgren07} for the principal ideal problem. The condition as derived in \cite{hallgren07} may not be satisfied for some infrastructures. Therefore, we relax this constraint and  clarify certain
crucial assumptions on the size of the preimage 
in Lemma~\ref{lm:fib_gN_size}.

\begin{lemma}\label{lm:fib_gN_size}
Let $\cA$ and $\cB$ be as in Lemma~\ref{lm:g_NKer}.  Consider the probability distribution $p=(p_{y,\ell})$ on $g_N(\cA\times\cB)$ where the probabilities of the elementary events $(y,\ell)$ are given by
\begin{equation}
p_{y,\ell} = \frac{|g_N^{-1}(y,\ell)|}{|\cA| |\cB|}\,.
\end{equation}
Let $X$ be the random variable that takes on the value $|g_N^{-1}(y,\ell)|$ if the event $(y,\ell)$ occurs.
Then, we have
\begin{equation}
\Pr\big( X \ge \kappa |\cA| \big) \ge \frac{1}{1-\kappa}\left(\frac{|\cB|}{\floor{R(N+1/d_{\min})}} - \kappa \right)
\end{equation}
for any $\kappa\in (0,1)$.
\end{lemma}
%
%
\proof{The expected value $\mathbb{E}[X]$ is bounded from below by
\begin{eqnarray*}
\mathbb{E}[X] 
& = & 
\sum_{(y,\ell)} p_{y,\ell} \, |g_N^{-1}(y,\ell)| \\
& = & 
|\cA| |\cB| \sum_{(y,\ell)} p_{y,\ell}^2 \\
& \ge &
|\cA| |\cB| \, \frac{1}{|g_N(\cA\times\cB)|} \\
& \ge &
|\cA| |\cB| \, \frac{1}{\floor{R(N+1/d_{\min})}}\,. 
\end{eqnarray*}
We used that the sum $\sum p_{y,\ell}^2$ is minimized when probability distribution is uniform on $g_N(\cA\times\cB)$ and $|g_N(\cA\times\cB)|\le\floor{R(N+1/d_{\min})}$.

Let $t=\Pr(X\ge \kappa\, \mathbb{E}[X])$. Then, we must have 
\begin{equation}
t|\cA| + (1-t)\kappa|\cA| \ge \mathbb{E}[X]\ge |\cA|\,|\cB| / \floor{R(N+1/d_{\min})} \nonumber
\end{equation} 
since $X$ is bounded by $|\cA|$ from above.  The desired lower bound on $t$ follows now easily.
}

%
%

\begin{theorem}\label{th:dlog}
Let $\mc{I}$ be an infrastructure containing at least 3 elements and satisfying the axioms {\bf{A}1--\bf{A}7}. 
For all $x\in X$, Algorithm~\ref{alg:dlog} returns an integer $\hat{d}_x$ such that $|d_x-\hat{d}_x|\leq 1$,
where $d_x$ is the distance of $x$.

Let $p_g$ be the probability of correctly evaluating $g_N$ and $\kappa$ a real number with
$(1-\sinc(3/4))/(1-2\sin(\pi/32))< \kappa <1-2/(2q+1)$.
Then, the success probability of the algorithm is $\Omega(1)$ and 
at least 
\begin{equation}
p_g \max_{\kappa} \left(1- \frac{2}{(2q+1)(1-\kappa)}\right)^2
\frac{\kappa^2}{2}\left(
1 - 2 \sin(\frac{\pi}{32})-  \frac{(1-\sinc(3/4))}{\kappa}
\right)^4 \!
\left(\frac{1}{64}-\frac{2}{B}\right)^2
\label{eq:dlogSuccess}
\end{equation}
where $B=\left\lfloor\hat{R} N\right\rceil $ and $q$ is chosen as in Algorithm~\ref{alg:dlog}.
\end{theorem}
%
%
\begin{algorithm}
\caption{{\ensuremath{\mbox{\scshape Generalized discrete logarithm.}}}
}\label{alg:dlog}
\begin{algorithmic}[1]
\STATE Choose $M\geq \ceil{2 R+1}$.
\STATE Determine $\hat{R}$ and $N$ such that $\left|M \left\lfloor\hat{R} N\right\rceil - M R N \right|\le 1/2$ and $N=q\ceil{2/d_{\min}}$ for a positive integer with $q\leq 4M$.  Set
$B=\fc{\hat{R}N}$ and $A=M B$.
\STATE Choose $L=\ceil{2 A \bar{k}\ceil{\hat{R}/d_{\bar{k}}}/(1-p_g)}N$.
\STATE Evaluate $g_N$ in superposition over $\{0,1\ldots,A-1\} \times \{0,1,\ldots,B-2\}$ twice.
\STATE Fourier sample over $\mbb{Z}_{A}\times 
\mbb{Z}_{B}$ to obtain $(h_1,k_1)$ and $(h_2,k_2)$.
\STATE Find integers $s,t$ such that $sk_1+tk_2=1$, using the extended Euclidean algorithm.
\STATE Compute $r=\frac{sh_1+th_2}{NM}$ 
\STATE Return $\hat{d}_x = r- \floor{r/\hat{R}}\hat{R}$.
\end{algorithmic}
\end{algorithm}
%
%
\begin{proof}
We compute an estimate $\hat{R}$ such that 
\begin{equation}
|R - \hat{R}|\le\epsilon \leq \frac{1}{16M^2\ceil{2/d_{\min}}}\,.
\end{equation}
We now show that there is an efficient method that determines positive integers $B=\fc{\hat{R} N}$ and $N$ such that
\begin{equation}
|M B - M  N R| \le \frac{1}{2}\,.
\end{equation}

To do this, we bound this deviation by 
\begin{eqnarray}
|M B - M  N R| &= &|M \lfloor N\hat{R}\rceil-M N \hat{R} +MN\hat{R} -M NR | \\
& \le & 
|M \lfloor N\hat{R}\rceil - M N \hat{R} | + M N \epsilon \label{eq:twoTerms}
\end{eqnarray}
The efficient method in Lemma~\ref{lm:cfbound} gives us a convergent $p/q$ with $q\leq 4M$ such that
\ben
\left|\frac{p}{q}-\hat{R}\ceil{2/d_{\min}} \right|\leq \frac{1}{4M q}\,. \label{eq:sizeB}
\een
The numerator $p$ has the form $\fc{\hat{R}\lceil 2/d_{\min} \rceil q}$.
The bound in equation~(\ref{eq:sizeB}) and the form of the numerator directly imply that $N=q\ceil{2/d_{\min}}$ has the desired properties.  Both terms in equation~(\ref{eq:twoTerms}) are smaller than $1/4$ for this choice.

Observe that $B-2 = \fc{\hat{R}N}- 2 \leq \lfloor R N \rfloor-1$ because $\hat{R}$ has been computed with such high precision.  We define the sets $\cB=\{0,1,\ldots,B-2\}$ and $\cA=\{0,1,\ldots,A-1\}$.

We create the superposition  
\begin{equation}
\frac{1}{\sqrt{|\cA| |\cB|}} 
\sum_{a\in\cA}
\sum_{b\in\cB} |a\rangle |b\rangle |g_N(a,b)\rangle\,.\nonumber
\end{equation}
We know that with probability greater or equal to $p_g$ all the values $g_N(a,b)$ are correct.

We measure the third register.   Denote the outcome by $(y,\ell)$. Lemma~\ref{lm:fib_gN_size} guarantees that $|\cA_{y,\ell}|\geq \kappa |\cA|$ holds
with probability greater or equal to
\begin{equation}
p_{\kappa}\geq \frac{1}{1-\kappa}\left(\frac{|\cB|}{\floor{R(N+1/d_{\min})}} - \kappa \right).
\end{equation}
Since  $N= q\ceil{2/d_{\min}}$, we can bound $p_{\kappa}$
\ben
p_{\kappa}&\geq& \frac{1}{1-\kappa}\left(\frac{NR-3}{{NR(1+1/2q)}} - \kappa \right), \\
&\geq& \frac{1}{1-\kappa}\left(\frac{2q-1}{2q+1}-\kappa\right) = 1- \frac{2}{(2q+1)(1-\kappa)},
\een
where we used the assumption that $\mc{I}$ has at least 3 elements and therefore  $R>3d_{\min}$, and $NR> 6q$.

Lemma~\ref{lm:g_NKer} implies that the post-measurement state has the form 
\begin{equation}
\frac{1}{\sqrt{|\cA_{y,\ell}|}} \sum_{a\in\cA_{y,\ell}} |a\> |b_a\>
\end{equation}
and there exists a unique $b_a$ for each $a\in\cA_{y,\ell}$ such that
\begin{equation}
b_a =  
-a d_x N + \left\lfloor \frac{ad_x+\frac{b_a}{N}+\frac{j}{L}}{R} \right\rfloor R N +
d_y N + \gamma_a N + \ell - \frac{j N}{L} \,,
\end{equation}
where $1/L \le \gamma_a \le 1/N-1/L$.  We rewrite the condition on $b_a$ as 
\begin{equation}
b_a = 
-a d_x N + \left\lfloor \frac{ad_x+\frac{b_a}{N}+\frac{j}{L}}{R} \right\rfloor R N +
\gamma_a N + \Delta\,,
\end{equation}
where $\Delta=d_y N+\ell-jN/L$ is constant. 

%
%

We apply the quantum Fourier transform over $\mbb{Z}_A\times \mbb{Z}_B$ to the first registers and obtain the superposition 
\begin{equation}
\frac{1}{\sqrt{A}} 
\frac{1}{\sqrt{B}}
\sum_{h\in\cA}
\sum_{k=0}^{B-1} 
\frac{1}{\sqrt{|\cA_{y,\ell}|}}
\sum_{a\in\cA_{y,\ell}}
\omega_A^{a h+Mb k}\ket{h} \ket{k}\,.
\end{equation}
The amplitude of the term $\ket{h}\ket{k}$ is given by 
\begin{equation}
\frac{1}{\sqrt{A}} 
\frac{1}{\sqrt{B}}
\frac{1}{\sqrt{|\cA_{y,\ell}|}}
\sum_{a\in\cA_{y,\ell}}
\omega_A^{a h + M b_a k} \,.
\end{equation}
The exponent  of $\omega_A$ in the previous equation is 
\begin{equation}
a h + M k \left( -ad_xN  + \floor{ \frac{ad_x + \gamma_a +  \frac{b_a}{N} + j/L}{R}} NR + \gamma_aN + \Delta \right)\,.
\end{equation}
The term $M k \Delta$ is independent of $a$ and can be dropped from the exponent since it does not change the probability distribution. 

We now show that we obtain a sample $(h,k)$ such that 
\begin{equation}\label{eq:niceCd}
h = kd_x M N  - \floor{\frac{kd_x}{R}} MN R  + \epsilon_h \mbox{ with $|\epsilon_h|\le \frac{1}{2}$}
\end{equation}
holds with high probability.\footnote{The reason that we consider
samples that have this particular form is as follows. 
Rearranging the terms in the exponent we see that
the sum is dominated by the terms $ah-(kd_x/R) MNR + k\floor{(ad_x+\gamma_a+b_a/Nj/L)/R}MNR$.
The exponent can be approximated as $ah-(kd_x/R -\floor{ad_x/R})MNR$.
Therefore, the probability of $(h,k)$ which is determined by the  geometric sum 
\be
\frac{1}{AB\cA_{y,\ell}}\left|\sum_{a\in \cA_{y,\ell}} \omega^{ah+Mb_a k }\right|^2 \approx 
\frac{1}{AB\cA_{y,\ell}}\left|\sum_{a\in \cA_{y,\ell}} \omega^{a(h -(\frac{kd_x}{R}-\floor{\frac{kd_x}{R}})MNR  }\right|^2
\ee
is large when $h= (kd_x/R-\floor{kd_x/R})MNR+\epsilon_h$, where $\epsilon_h$ is to ensure that 
$h$ is an integer.} 

As shown previously, $N$ is chosen such that $M  N R - M \left\lfloor N \hat{R} \right\rceil = \eta$ with $|\eta| \le \frac{1}{2}$.  
To simplify the notation, we use $x$ to denote the distance $d_x$ of the element $x$ throughout the rest of the proof.
The exponent of $\omega_A$ modulo $A$ is 
\begin{eqnarray}
&   & 
\phantom{+}
a \left( k x M N  - \floor{\frac{kx}{R}} M NR  + \epsilon_h \right) \nonumber \\
&   & 
+ M k \left( -axN  + \floor{ \frac{ax + \gamma_a +  \frac{b_a}{N} + j/L}{R}} NR + \gamma_aN \right) \nonumber\\
& = & 
\left(
k \ceil{ \frac{ax + \gamma_a +  \frac{b_a}{N} + j/L}{R}} - a \floor{\frac{kx}{R}}
\right) M  N  R+
\epsilon_h a + M N \gamma_a k \nonumber\\
& \equiv &
\eta \left(
k \ceil{ \frac{ax + \gamma_a +  \frac{b_a}{N} + j/L}{R}} - a \floor{\frac{kx}{R}}
\right) +
\epsilon_h a + M N \gamma_a k  \nonumber\\
& = & 
\eta \left(
k \left(\frac{ax}{R} + \delta_a \right) - a \left( \frac{kx}{R}+\zeta\right) \right) +
\epsilon_h a + M k \gamma_a \nonumber \\
& = &
\eta \delta_a k  - \eta \zeta a + \epsilon_h a + M N \gamma_a k \nonumber\\
& = & \delta a +\theta_a \,.
\end{eqnarray}
The (constant) factor $\delta:=\epsilon_h- \eta \zeta$ in front of $a$ is less than $3/4$ in absolute value ($\epsilon_h\le\frac{1}{2}$, $\zeta<1$ and $\eta<\frac{1}{2}$).  
Assume we measure $k\le \floor{B/64}-1$.  Then, for each $a$ the term $\theta_a:=(\eta \delta_a + M N\gamma_a) k$ is less than $A/32$ in absolute value (since $|\delta_a| < 2$ and $|\gamma_a N| < 1$). 

We can now apply Lemma~\ref{lm:pertMissingSeries} to bound the probability of measuring $(h,k)$ as in equation~(\ref{eq:niceCd}); we denote
this probability by $p_{hk}$.  Note that $A$ corresponds to $n$, the summation index $a$ to $j$ and $\cA_{y,\ell}$ to the set $\mc{J}$ in the Lemma~\ref{lm:pertMissingSeries}.  

The probability $p_{hk}$ is bounded from below by
\ben
p_{hk} &\geq& \frac{|\cA_{y,\ell}|}{A B}\left(
1 - 2 \sin(\pi/32)- \Big(
1-\sinc(3/4) 
\Big) \frac{1}{\kappa}
\right)^2\,,
\een
where $c_{\delta}$ is as in Lemma~\ref{lm:pertMissingSeries}. 

The probability of any good pair $(h,k)$ (with the restriction $k\le \floor{B/64}-1$) is bounded from below by
\begin{equation}
{\kappa }\left(
1 - 2 \sin(\pi/32)- \Big(
1-\sinc(3/4)  
\Big) \frac{1}{\kappa}
\right)^2 \left(\frac{1}{64}-\frac{2}{B}\right)\,,
\end{equation}
where we used that $|\cA_{y,\ell}|\geq \kappa |A|$ and $\floor{B/64-1}\geq B/64-2$.

%
%

We now show how to obtain an estimate of the distance of $x$ from two good pairs $(h_1,k_1)$ and $(h_2,k_2)$ with the additional restriction that 
$k_1,k_2$ are coprime.  This is based on the method in \cite{hallgren07}. We have
$h_i = k_i x N M - \floor{{k_i x}/{R}} R N M + \epsilon_i$ with $|\epsilon_i|\le\frac{1}{2}$. Since $k_1,k_2$ are coprime we know there exist integers $s,t$ such that 
$sk_1+tk_2=1$, which can be computed by the extended Euclidean algorithm.
Let $r=({sh_1+th_2})/{MN}$, then we have 
\begin{eqnarray*}
\frac{sh_1+th_2}{MN} 
& = & 
s k_1 x  - s\floor{\frac{k_1x}{R}} R   + \frac{s\epsilon_1}{MN} +tk_2 x   - t\floor{\frac{k_2x}{R}} R   + \frac{t\epsilon_2}{MN} \\
& = & 
(sk_1+tk_2) x- s\floor{\frac{k_1x}{R}} R   - t\floor{\frac{k_2x}{R}} R + \frac{s\epsilon_1  + t\epsilon_2 }{MN}\\
& = & 
x   -s\floor{\frac{k_1x}{R}} R   - t\floor{\frac{k_2x}{R}} R + \frac{s\epsilon_1 + t\epsilon_2}{MN}\\
&=& x - mR+\epsilon_r, 
\end{eqnarray*}
where $\epsilon_r=({s\epsilon_1 + t\epsilon_2})/{MN}$.
Since $|s|,|t|\leq \max\{k_1,k_2 \}$, and $k_1,k_2\leq \left[\hat{R}N\right]/32$, it follows that 
$\epsilon_r=\frac{s\epsilon_1 +  t\epsilon_2}{MN} \leq \frac{\floor{RN}}{MN} <1/2$ by our choice of  $M$. Furthermore, $|m|\leq NR/8$, as 
 $|r|\leq 2M[N\hat{R}] [N\hat{R}]/32MN < NR^2/8$.

We can estimate $x$ by reducing $r$ modulo $\hat{R}$
to bring it within the range $[0,\hat{R})$. This gives us an estimate 
$ \hat{x} = x-m(R-\hat{R}) +\epsilon_r$ and the error $|x-\hat{x}|$ can be bounded as follows:
\be
|x-\hat{x}| &\leq & |m(R-\hat{R})|+\epsilon_r \leq| m \epsilon| +
|\epsilon_r|\\
&\leq &\frac{{NR}}{8} \frac{1}{16M^2 \ceil{2/d_{\min}}}+|\epsilon_r|\leq1,
\ee
where we used the fact that $M>2R$ and $N\leq 4M \ceil{2/d_{\min}}$ and $|\epsilon_r|< 1/2$.

The probability of measuring two good samples   $(h_1,k_1)$ and $(h_2,k_2)$ such that 
$k_1,k_2$  are coprime is  given by
\ben
p_{\rm{success}} \geq \frac{1}{2}(p_{\kappa}p_{hk}(1/64-2/B))^2 p_g,
\een
where $p_g$ is the probability of evaluating $g_N$ successfully.
\end{proof}

We make the following observations regarding the success probability of the quantum algorithm.
First, a simpler lower bound on the success probability can be obtained without having to maximize
over $\kappa$ in equation~\eqref{eq:dlogSuccess}, by evaluating this expression at 
$\kappa=(\kappa_1+\kappa_2)/2$, where 
$\kappa_1= (1-c_\delta)/(1-2\sin(\pi/32))$ and $\kappa_2= 1-2/(2q+1)$.
We also note the expression can be further simplified to be completely independent of 
of the size of the infrastructure as follows. 

Second, under the assumption that $R\geq 256$
and $q\geq 8$, 
we can bound $(1/64-2/B) \geq 1/128$ and $2/(2q+1) \leq 1/8$, and the lower bound on 
success probability simplifies to a constant independent of the problem size. 
\ben
\max_{\kappa} p_g \left(1- \frac{1}{8(1-\kappa)}\right)^2
\frac{\kappa^2}{2}\left(
1 - 2 \sin(\frac{\pi}{32})-  \frac{(
1-\sinc(3/4)
)}{\kappa}
\right)^4
\left(\frac{1}{128}\right)^2 \,\label{eq:dlogSuccessSimp}
\een
Although the expressions for the success probability may appear to be a little unwieldy, 
we hope they provide insight into the various factors affecting the success probability.

Third, we can boost the success probability (strictly speaking, the lower bound on it) by increasing $q$. 

Fourth, we can truly improve upon the success probability by extending the set of usable observations
$(h_1,k_1)$ and $(h_2,k_2)$. Currently, we require that $k_i< \floor{B/64}$, but this can 
be relaxed significantly. 

\section*{Acknowledgment}

We would like to thank Felix Fontein for helpful discussions on infrastructures 
and suggestions to improve the paper.  P.W. thanks Joseph Brennen, Chen-Fu Chiang, and 
(Raymond) Yiu Yu Ho for helpful discussions. 
P.S. thanks  Robert Raussendorf for his generous support. 

P.W. gratefully acknowledges the support from the NSF grant CCF-0726771 and the NSF CAREER Award CCF-0746600. P.S. was sponsored by grants from CIFAR, MITACS and NSERC.

\section*{Appendix}

We prove here some auxiliary results.

\begin{lemma}\label{lm:coprimeProb}
Let $a$ and $b$ be two random numbers chosen uniformly at random from $\{1,\ldots,N\}$.  The probability that $a$ and $b$ are coprime is bounded from below by $1/2$, i.e.,
\begin{equation}
\Pr(\gcd(a,b)=1) > \frac{1}{2}\,.
\end{equation}
\end{lemma}
\begin{proof}
Let $p$ be an arbitrary prime.  Then the probability that $p$ divides $a$, denoted $\Pr(p\mid a) $, 
is given by 
\begin{eqnarray}
\Pr(p \mid a) = \frac{\lfloor \frac{N}{p} \rfloor}{N} \le \frac{1}{p}\,.\nonumber
\end{eqnarray}
Thus, 
\begin{equation}
\Pr(p \mid \gcd(a,b) ) \le \frac{1}{p^2}\,.\nonumber
\end{equation}
We obtain an upper bound on the probability that there is a prime dividing the greatest common divisor of $a$ and $b$ with the help of the union bound.  This yields  
\begin{equation}
\Pr(\gcd(a,b)>1) \le \sum_p \frac{1}{p^2}\,,\nonumber
\end{equation}
where the summation index $p$ ranges over all primes.  The sum of squared reciprocals of primes is known to be
\begin{equation}
\sum_p \frac{1}{p^2} = \sum_{k=1}^\infty \frac{\mu(k)}{k} \ln \zeta(2k) = 0.4522474200\ldots\,,\nonumber
\end{equation}
where $\mu$ denotes the M{\"o}bius mu function and $\zeta$ the Riemann zeta function \cite[page 95]{Finch}.  Finally, we obtain the desired result
\begin{equation}
\Pr(\gcd(a,b)=1) \ge 1 - \sum_p \frac{1}{p^2} > \frac{1}{2}\nonumber
\end{equation}
by considering the complementary event.
\end{proof}

We now prove a result related to continued fractions. The reader can find more details about
continued fractions in \cite{Burton10}.
\begin{lemma}\label{lm:cfbound}
Let $p_i/q_i$  denote the convergents of a real number $r\in \mbb{R}$, for $i\in \mbb{N}$. 
Then for any given constant $c>1$, there exists a convergent $p_{\ell}/q_{\ell}$ such that 
$|r-p_{\ell}/q_{\ell}|<1/c q_{\ell}$ and $q_{\ell} \leq c$.
\end{lemma}
\begin{proof} 
Since $c>1=q_0$ and  $q_i$ form a monotonically
increasing sequence for $i>1$, there exists  such a convergent 
$p_\ell/q_\ell$ such that  $q_{\ell}\leq c<q_{\ell+1}$ unless $r$
has a finite continued fraction expansion with all the $q_i<c$. If the latter case occurs, then it follows that there exists a convergent $p_{\ell}/q_{\ell} $ such that $r=p_{\ell}/q_{\ell}$ therefore for this convergent 
$|r-p/q|=0<1/c $ and the statement of the lemma holds. Otherwise,  $r$ has a continued fraction expansion such that $q_{\ell}\leq c< q_{\ell+1}$.  We
know that the convergents satisfy the relation 
\be
\left|r- \frac{p_i}{q_i}\right| <\frac{1}{q_iq_{i+1}}.
\ee
Therefore, we must have 
\be
\left|r- \frac{p_\ell}{q_\ell}\right| <\frac{1}{q_{\ell}q_{\ell+1}}<\frac{1}{cq_{\ell}},
\ee
where we used the fact that $q_{\ell+1}>c$.
\end{proof}


\def\cprime{$'$}

\end{document}